\documentclass[12pt,draftclsnofoot,journal,onecolumn]{IEEEtran}
\usepackage{mathtools}
\usepackage{ifpdf}
\usepackage{amsthm}
\usepackage{amsmath}
\usepackage{nicefrac}
\usepackage{cite}
\usepackage{amsfonts}
\usepackage{comment}
 \usepackage{url}
 \usepackage{amssymb}
 \usepackage[paper=letterpaper,margin=0.75in]{geometry}
 \usepackage[ansinew]{inputenc}

\usepackage{lipsum}
\usepackage{hyperref}

\newtheorem{thm}{Theorem}
\newtheorem{lem}{Lemma}
\theoremstyle{definition}

\newcommand*{\QEDB}{\hfill\ensuremath{\square}}%

\makeatletter
\renewenvironment{proof}[1][\proofname] {\par\pushQED{\qed}\normalfont\topsep6\p@\@plus6\p@\relax\trivlist\item[\hskip\labelsep\bfseries#1\@addpunct{.}]\ignorespaces}{\popQED\endtrivlist\@endpefalse}
\makeatother

\begin{document}
\title{Covert Communications on Renewal Packet Channels}

\author{
   \IEEEauthorblockN{Ramin Soltani\IEEEauthorrefmark{1},
            Dennis Goeckel\IEEEauthorrefmark{1}, Don Towsley\IEEEauthorrefmark{2}, and Amir Houmansadr\IEEEauthorrefmark{2}}

\IEEEauthorblockA{\IEEEauthorrefmark{1}Electrical~and~Computer~Engineering~Department,~University~of~Massachusetts,~Amherst,
    \{soltani, goeckel\}@ecs.umass.edu\\}
    \IEEEauthorblockA{\IEEEauthorrefmark{2}College of Information and Computer Sciences, University of Massachusetts, Amherst,
    \{towsley, amir\}@cs.umass.edu}
        
                       \thanks{ This work has been supported by the National Science Foundation under grants ECCS-1309573 and CNS-1525642.}
                       \thanks{ This work has been presented at the 54th Annual Allerton Conference on Communication, Control, and Computing,  October 2016.}
                       \thanks{Personal use of this material is permitted. Permission from IEEE must be obtained for all other uses, in any current or future media, including reprinting/republishing this material for advertising or promotional purposes, creating new collective works, for resale or redistribution to servers or lists, or reuse of any copyrighted component of this work in other works. DOI: \href{https://doi.org/10.1109/ALLERTON.2016.7852279}{10.1109/ALLERTON.2016.7852279}}

}

\date{}
\maketitle
\thispagestyle{plain}
\pagestyle{plain}

\newtheorem{definition}{Definition}

\begin{abstract}
Security and privacy are major concerns in modern communication networks. In recent years, the information theory of covert communications, where the very presence of the communication is undetectable to a watchful and determined adversary, has been of great interest. This emerging body of work has focused on additive white Gaussian noise (AWGN), discrete memoryless channels (DMCs), and optical channels. In contrast, our recent work introduced the information-theoretic limits for covert communications over packet channels whose packet timings are governed by a Poisson point process. However, actual network packet arrival times do not generally conform to the Poisson process assumption, and thus here we consider the extension of our work to timing channels characterized by more general renewal processes of rate $\lambda$. We consider two scenarios. In the first scenario, the source of the packets on the channel cannot be authenticated by Willie, and therefore Alice can insert packets into the channel. We show that if the total number of transmitted packets by Jack is $N$, Alice can covertly insert $\mathcal{O}\left(\sqrt{N}\right)$ packets and, if she transmits more, she will be detected by Willie. In the second scenario, packets are authenticated by Willie but we assume that Alice and Bob share a secret key; hence, Alice alters the timings of the packets according to a pre-shared codebook with Bob to send information to him over a $G/M/1$ queue with service rate $\mu>\lambda$. We show that Alice can covertly and reliably transmit $\mathcal{O}(N)$ bits to Bob when the total number of packets sent from Jack to Steve is $N$.
\end{abstract}


\section{Introduction}
Secrecy and privacy are key concerns in modern communication systems. Most security research is focused on protecting the content of the message from being decrypted by an adversary, and there has been significant work in both traditional cryptographic approaches and information-theoretic secrecy approaches to achieve this goal. However, as clearly illustrated in recent high-profile security scenarios (e.g. the Snowden disclosures \cite{snowden}), it is often the mere presence of a message between two parties that must be hidden rather than just its content. Applications range from military scenarios, where the volume of radio traffic can indicate the presence and magnitude of activity, to domestic scenarios, where an authority might punish certain parties or at least shut down any communications, particularly those that are encrypted.

Stenography \cite{ker07pool} is a solution for hiding the existence of communication by embedding secret content in an overt message on a digital channel that is generally noiseless. And spread spectrum methods have been used for many years to provide covert communication in noisy continuous-valued channels. However, the information-theoretic limits of covert communication were only recently studied for additive white Gaussian (AWGN) channels \cite{bash12sqrtlawisit,bash_jsac2013,korzhik2005existence} and later extended to provide a comprehensive characterization of the limits of covert communication over discrete memoryless channels (DMCs), optical channels, and AWGN channels \cite{boulat_commmagg, che13sqrtlawbscisit,hou14isit,bloch15covert-isit,wang15covert-isit}. Hence, this is an active and rapidly growing area of information-theoretic research.

In contrast to the bulk of the work in this emerging area, our work in~\cite{soltani2015covert} considered covert communication over packet-based channels where the timing of the packets is modeled by a Poisson point process. However, in practice, many channels do not have packet timings that obey such a convenient model. Hence, here we extend our results from \cite{soltani2015covert} to scenarios where the packet timings are governed by a more general renewal process. As in \cite{soltani2015covert}, we will exploit the pioneering work of Anantharam and Verdu~\cite{verdubitsq} on the information-theoretic limits when communicating with packet timing through a queue; however, in contrast to our work in \cite{soltani2015covert}, the results from \cite{verdubitsq} must undergo non-trivial modification to fit the $G/M/1$ model introduced as part of our construction for covert communications. 

In particular, we consider a channel where an authorized (overt) transmitter Jack sends packets to an authorized (overt) receiver Steve, where the timings of packet transmission are modeled by a renewal (point) process with inter-arrival times governed by a probability density function (pdf) $p_0(x)$ and rate $\lambda=\left(\int_{x=0}^{\infty}x p_0(x)\right)^{-1}$ packets per second. Covert transmitter Alice wishes to transmit data to a covert receiver, Bob, on this channel in the presence of an attentive adversary, Willie, who is monitoring the channel between Alice and Bob precisely to detect such transmissions. We consider two scenarios in detail.

In the first scenario, we assume: (1) the warden Willie is not able to see packet contents, and therefore cannot authenticate the source of the packets (e.g., whether they are actually sent by Jack); and (2) Alice is restricted to packet insertion. Willie is aware that the timing of the packets of the allowed (i.e. overt) communication link follows a renewal process with inter-arrival time pdf $p_0(x)$, so he seeks to apply hypothesis testing to verify whether the packet process has the proper characteristics.

In \cite{soltani2015covert}, the inter-arrival time was exponential, and thus for the packet insertion we were able to exploit the fact that the superposition of two independent Poisson point processes is a Poisson point process; Alice simply generated a Poisson point process of the appropriate rate and used it to govern the timings of her packet insertions onto the Jack-to-Steve channel. However, such a technique does not readily extend to channels governed by non-Poisson renewal processes, and thus a different technique is required here. In particular, Alice will generate a renewal process with a slightly higher rate than that of Jack by scaling $p_0(x)$. This allows Alice to transmit covert packets at a low rate along with Jack's transmitted packets at rate $\lambda$. For a given packet timing generated from her (slightly) faster renewal process, she will decide whether she should send a covert or overt packet by generating a Bernoulli random variable with a low probability of ``success'', where ``success'' results in the transmission of a covert packet. However, a complication arises, as this approach requires that Alice always have overt packets available to send when indicated. Therefore, she first buffers some overt packets. In particular, Alice employs a two-phased system. In the first phase, she will (slightly) slow down the transmission of packets from Jack so as to build up a backlog of packets in her buffer. In the next phase, she then generates the renewal process with a slightly higher rate and sends both covert and overt packets.

The first result is established by analyzing the two phases for covertness, where covertness is defined formally as in \cite{bash_jsac2013}: if $\mathbb{P}_{FA}$ is the probability of false alarm at Willie's detector and $\mathbb{P}_{MD}$ is his probability of missed detection, a scheme is covert if Willie's sum of error probabilities $\mathbb{P}_{FA} + \mathbb{P}_{MD} > 1-\epsilon$ for any $0<\epsilon<1$.  First, in Lemma~\ref{lem:2}, we show that Alice can collect and store $\mathcal{O}\left(\sqrt{N}\right)$ packets in a packet stream of length $N$ transmitted by Jack during the first phase while being covert; conversely, if she collects more, she will be detected by Willie with high probability. Then, we show that (see Theorem~3) if Alice decides to transmit $\mathcal{O}\left(\sqrt{N}\right)$ packets to Bob during the second phase, where $N$ is the total number of packets sent by Jack, she will remain covert. A crucial part of this proof is showing that Alice has buffered enough packets during the first phase so as to not run out of overt packets during the second phase. Finally, conversely, we prove that if Alice transmits ${\omega}\left(\sqrt{N}\right)$ packets, she will be detected by Willie with high probability.

In the second scenario, we assume that Willie can look at packet contents and therefore can verify packets' authenticity. Thus, Alice is not able to insert packets, but we allow Alice a secret key and the ability to alter the packet timings to convey information to Bob, whom is receiving the packets through a $G/M/1$ queue with service rate $\mu> \lambda$.  To do such, Alice designs an efficient code, where a codeword consists of a sequence of packet timings drawn from the same process as the overt traffic; hence, a codeword transmitted with those packet timings is undetectable. However, there is a causality constraint, as Alice clearly cannot send the next packet (i.e. codeword symbol) unless she has a packet from the Jack to Steve link available to transmit.  This suggests the following two-stage process. In the first stage, Alice covertly slows down the transmission of the packets from Jack to Steve so as to buffer some number of packets. In the second stage, Alice continues to add packets transmitted by Jack to her buffer while releasing packets with the inter-packet delay appropriate for the chosen codeword.

Alice's scheme breaks down when her buffer is empty at any point before completing the codeword transmission. Hence, the question becomes: how long must Alice collect packets during the first stage so as to guarantee (with high probability) that she will not run out of packets before the completion of codeword transmission during the second stage? First, in Lemma~\ref{lem:3}, we show that Alice can achieve a positive capacity for the $G/M/1$ queue if she embeds information in the packet timings in this fashion. Building on Lemma~\ref{lem:2} and Lemma~\ref{lem:3}, we prove (Theorem~4) that, using our two-stage covert communications approach, Alice can reliably and covertly transmit $\mathcal{O}(N)$ bits in a packet stream of length $N$.

The remainder of the paper is organized as follows. In Section~\ref{modcons}, we present the system model, definitions, and metrics. Then, we review the results for Poisson packet channels in Section III and provide constructions and their analysis for non-Poisson channels in Section IV. Section V contains the discussion and section VI summaries our conclusions.
\section{System Model, Definitions, and Metrics}
\label{modcons}

\subsection{System Model} \label{sec:1}
Suppose that Jack transmits packets to Steve, while a watchful warden Willie observes the packets flowing from Jack to Steve and attempts to discern any irregularities that might indicate someone is altering aspects of the packet stream to convey information. Indeed, Alice's goal is to do exactly that: manipulate the packets sent by Jack to Steve so as to communicate covertly with Bob, who is located beyond the warden Willie but before the intended recipient Steve. One such scenario illustrating the location of the various parties is shown in Fig.~\ref{fig:SysMod}. We consider the two specific operating scenarios for this problem.

{\em Scenario 1 (Packet insertion):} In Scenario 1, which is shown in Fig. (\ref{fig:SysMod}) and analyzed in Section IV.A, we assume that: 
\begin{enumerate}
\item Transmission times for the packets transmitted by Jack are modeled by a renewal process in which the inter-arrival times are positive i.i.d random variables with probability density function (pdf) $p_0(x)$ and transmission rate is $\lambda=\left(\int_{0}^{\infty}x p_0(x)\right)^{-1}$. We will term this a ``renewal channel''.
\item Willie is not able to authenticate the packets to see if a packet is coming from Jack.
\item Alice, with knowledge of $p_0(x)$, is allowed to insert and transmit her own packets, buffer and release Jack's transmitted packets when she desires,  but not share a codebook with Bob.
\item Bob is able to authenticate, receive and remove the covert packets; therefore, Steve does not observe the covert packets.
\item Willie knows that the legitimate communication is modeled by a renewal process with inter-arrival time pdf $p_0(x)$, and he knows all of the characteristics of Alice's packet buffering and release scheme. 
\end{enumerate}

\noindent In this scenario, we determine the number of packets that Alice can insert covertly into the channel while remaining covert. \QEDB

{\em Scenario 2 (Packet timing):} In Scenario 2, which is shown in Fig.~\ref{fig:SysMod2} and analyzed in Section IV.B, we assume that: 
\begin{enumerate}
\item Packet transmission times are modeled by a renewal process (as in Scenario 1).
\item Willie is able to access packet contents and hence can authenticate whether a packet comes from Jack. Therefore, Alice cannot insert packets into the channel.
\item Alice and Bob can share a secret codebook based on which Alice alters the packet timings by buffering packets and releasing them when she desires into the channel, thereby enabling covert communication through packet timing control.
\item Bob has access to the resulting packet stream only after it passes through a queue which
\begin{itemize}
\item processes the packets on a First-In-First-Out (FIFO) basis, i.e, upon departure of a packet, the next packet waiting in queue is processed.
\item has i.i.d exponential service times. 
\item has a service rate of $\mu > \lambda$.
\item is in equilibrium, and its input and output rate are equal.
\end{itemize}
\item Willie knows that the legitimate communication is modeled by a renewal process with inter-arrival time pdf $p_0(x)$, and he knows all of the characteristics of Alice's packet buffering and release scheme except a secret key that is pre-shared between Alice and Bob. 
\end{enumerate}

\noindent In this scenario, we calculate the number of bits that Alice can reliably and covertly transmit to Bob without detection by Willie. \QEDB
\subsection{Definitions}
\label{defs}

For the queue in Scenarios 2, denote the $i^{th}$ inter-arrival time and inter-departure time between the $i^{th}$ and $(i+1)^{th}$ packet by $A_i$ and $D_i$ respectively, where $1\leq i \leq n$. Therefore, $\sum_{j=1}^{i}A_j$ and $\sum_{j=1}^{i}D_j$ are the arrival and departure times of the $i^{th}$ packet respectively. Also, denote the service and idling time for the $i^{th}$ packet by $S_i$ and $W_i$ respectively. Note that $D_i=W_i+S_i$, and $W_i$ is the time between $(i-1)^{th}$ departure and $i^{th}$ arrival. Therefore,
\begin{align}
W_i = \max\bigg\{0,\sum\limits_{j=1}^{i}A_j - \sum\limits_{j=1}^{i-1}D_i\bigg\}
\end{align}

\begin{definition} \label{def:v1} \cite[Definition 1]{verdubitsq}: An $\left(n,M,T, \delta\right)$-code for a queue consists of a codebook of $M$ codewords, each of which is a vector of $n$ positive inter-arrival times $\{a_i\}_{i=1}^{n}$ such that the $k^{th}$ arrival occurs at $\sum_{i=1}^{k} a_i$; a decoder which upon observation of all $n$ departures selects the correct codeword with probability greater than $1- \delta$, assuming that the queue is in equilibrium. The $n^{th}$ departure from the queue occurs on the average (over equiprobable codewords and the queue distributions) no later than $T$. The rate of an $\left(n,M,T, \delta\right)$-code is defined as $\frac{\log M}{T}$.
\end{definition}
Note that in Definition 1 in \cite{verdubitsq}, the queue is initially empty. However, similar to \cite[Theorem 6]{verdubitsq}, this condition is replaced with the condition that the queue is in equilibrium in the above definition. Also, \cite[Definition 1]{verdubitsq} includes the condition that the inter-arrival times are non-negative. However, we have changed non-negative to positive since in all of our scenarios the inter-arrival times are positive.
\begin{definition} \label{def:rate} \cite[Definition 2]{verdubitsq}
$R$ is $ \delta-$achievable at output rate $\lambda$ if for all $\gamma>0$ there exists a sequence of $\left(n,M,n/\lambda, \delta\right)$-codes such that
\begin{align}
\lambda \frac{\log M}{n} > R-\gamma
\end{align}
Rate $R$ is achievable at output rate $\lambda$ if it is $ \delta$-achievable at output rate $\lambda$ for all $0<\delta<1$. The capacity of the queue at output rate $\lambda$, is the maximum achievable rate at output rate $\lambda$.
\end{definition}
\begin{definition} (Hypothesis Testing) Willie is faced with a binary hypothesis test: the null hypothesis ($H_0$) corresponds to the case that Alice does not transmit, and the alternative hypothesis $H_1$ corresponds to the case that Alice transmits. We denote the distributions of sequences of inter-arrival times that Willie observes by $\mathbb{P}_{1}$ and $\mathbb{P}_{0}$ under $H_1$ and $H_0$ respectively.

Also, we denote by $\mathbb{P}_{FA}$ the probability of rejecting $H_0$ when it is true (type I error or false alarm), and $\mathbb{P}_{MD}$ the probability of rejecting $H_1$ when it is true (type II error or missed detection). We assume that Alice's probability of transmission is $\frac{1}{2}$ and Willie knows that. Also, we assume that Willie uses classical hypothesis testing  and seeks to minimize $\mathbb{P}_{FA} + \mathbb{P}_{MD}$; the generalization to arbitrary prior probabilities is straightforward, see~\cite{bash_jsac2013}. 
\end{definition}
\begin{definition} (Covertness) Alice's transmission is {\em covert} if and only if she can bound Willie's average sum of probabilities of error $\mathbb{E}[\mathbb{P}_{FA}+\mathbb{P}_{MD}]$ by $1-\epsilon$ for any $\epsilon>0$~\cite{bash_jsac2013}.
\end{definition}
\begin{definition} (Reliability) A transmission scheme is {\em reliable} if and only if the probability that a codeword transmission from Alice to Bob is unsuccessful is upper bounded by $\zeta$ for any $\zeta>0$. Note that this metric applies in Scenario 2.
\end{definition}
\subsection{Metrics}
In this paper, a covert packet is a packet that is inserted by Alice into the channel (not originally from Jack), and an overt packet is a packet that is transmitted originally by Jack. We denote the number of covert packets that Alice can insert into the channel (in Scenarios 1) and the number of overt packets that Alice can buffer covertly (in Lemma~\ref{lem:2}) by $N_c$. Also, we denote the amount of covert information that Alice can convey to Bob through inter-packet delays, in Scenario 2 by $N_b$.

\begin{figure}
\begin{center}
\includegraphics[width=\textwidth/2 ,height=\textheight,keepaspectratio]{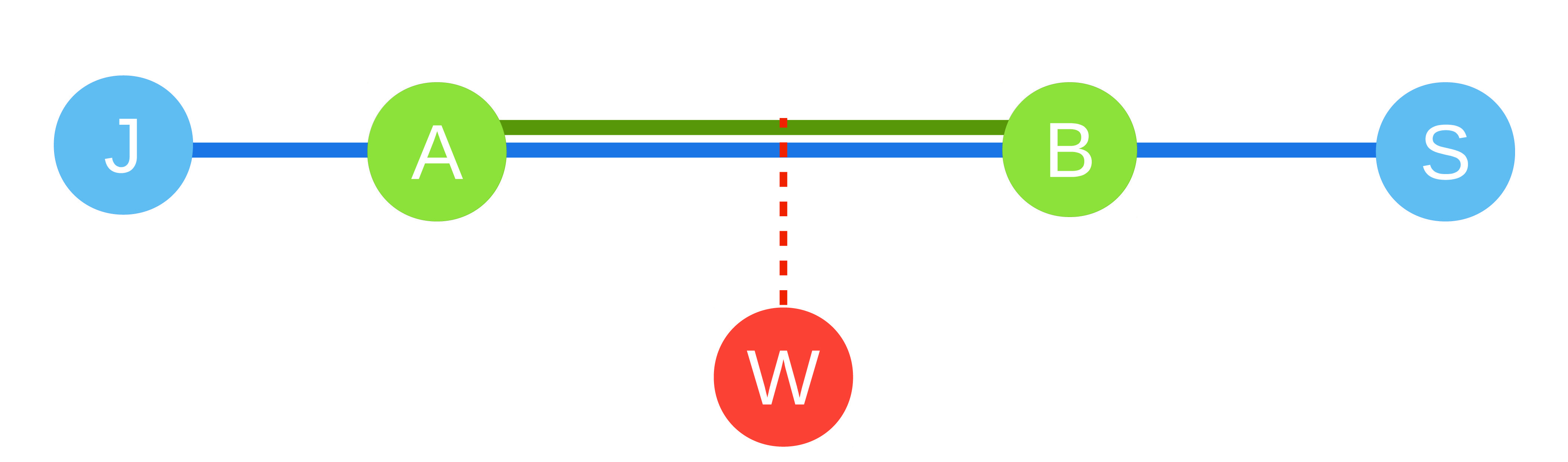}
\end{center}
 \caption{System configuration for Scenario~1: Willie cannot authenticate packets. Therefore, Alice inserts her packets into the channel between Jack and Steve to communicate covertly with Bob. The blue color shows the legitimate communication and the green shows the covert communication.}
 \label{fig:SysMod}
 \end{figure}

\begin{figure}
 \begin{center}
\includegraphics[width=\textwidth/2,height=\textheight,keepaspectratio]{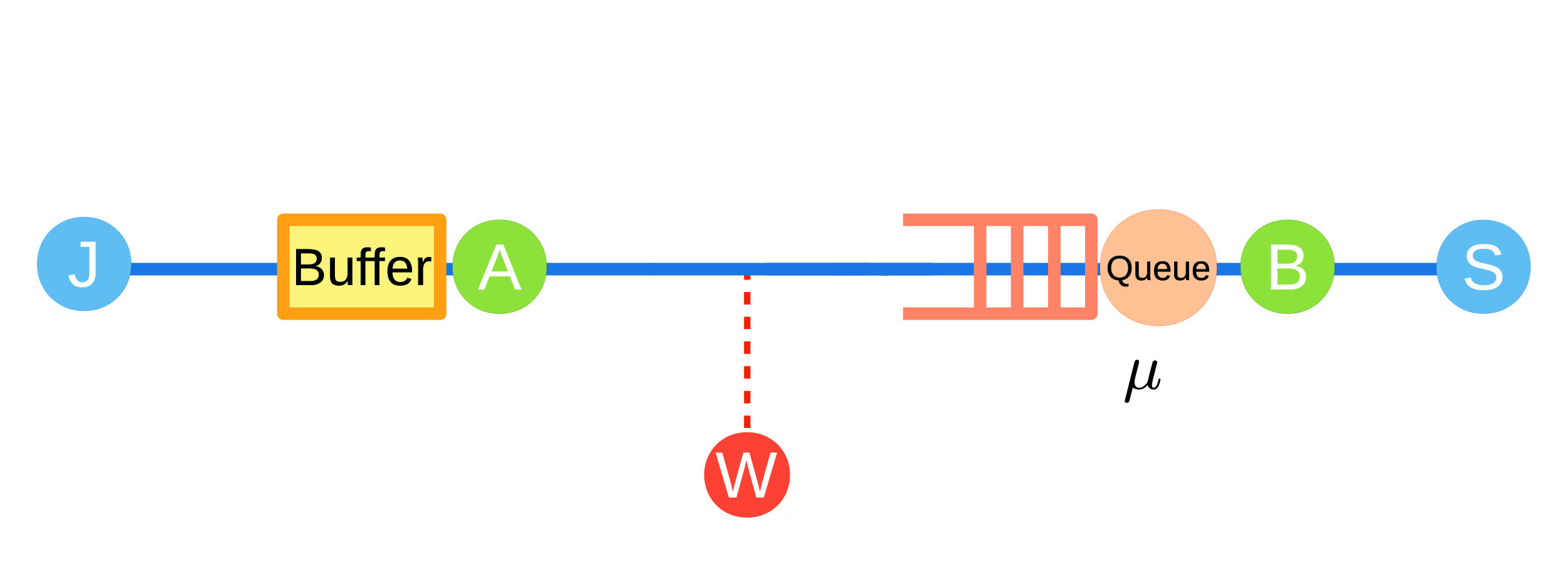}
 \end{center}
 \caption{System configuration for Scenario 2: 
Willie can authenticate packets so Alice embeds information in the packet timings. Alice is able to buffer packets in order to alter packet timings and Bob has access to the packet stream through an exponential server queue with service rate $\mu$.}
 \label{fig:SysMod2}
 \end{figure}
\section{Poisson Channels}
In this section, we review the results for Poisson channels. Consider the following Theorems~\cite{soltani2015covert}:
\begin{thm}
Consider Scenario 1. If the timings of the packets ar modeled by a Poisson point process of rate $\lambda$, and Alice is allowed to insert packets, she can covertly insert $\mathcal{O}(\sqrt{\lambda T})$ packets in a time interval of length $T$. Conversely, if Alice attempts to insert $\omega\left(\sqrt{\lambda T}\right)$ packets in a time interval of length $T$, there exists a detector that Willie can use to detect her with arbitrarily low sum of error probabilities $\mathbb{P}_{FA} + \mathbb{P}_{MD}$.
\end{thm}
\begin{thm} \label{th2}
Consider Scenario 2. If the timings of the packets ar modeled by a Poisson point process of rate $\lambda$, by embedding information in the inter-packet delays, Alice can covertly and reliably transmit $\mathcal{O}\left(\lambda T \right)$ bits to Bob in a time interval of length $T$. 
\end{thm}
See details of the proofs for each of the above Theorems including the communication schemes, construction and analysis in ~\cite{soltani2015covert}.

\section{Non-Poisson Channels}
As described in Section I, the packet arrival processes measured in many networks demonstrate non-Poisson behavior. Hence, in this section, we extend our results from section III to the non-Poisson case.
\subsection{General Renewal Model, Packet Insertion (Scenario 1)}
In this section, we consider Scenario 1: On a renewal channel, Willie cannot authenticate packets to see whether they are from Jack or Alice, and Alice is only allowed to send information to Bob by inserting packets into the channel. 

Per Section II, we assume that the inter-arrival times of the packets transmitted by Jack are i.i.d and their pdf is $p_0(x)$; thus Jack's transmission rate is $\lambda=\left(\int_{0}^{\infty}x p_0(x)\right)^{-1}$. For the transmission of covert packets, Alice generates a renewal process $B$ in which the pdf of the inter-arrival times is $p_1(x,\rho)=\frac{1}{1-\rho}p_0(\frac{x}{1-\rho})$, where $0<\rho<1$. Note that $p_1(x,\rho)$ is a scaled version of $p_0(x)$ that (slightly) lengthens the inter-arrival times, and therefore the rate of the generated renewal process $\lambda_B$ is (slightly) higher than Jack's transmission rate $\lambda$. This, allows Alice to transmit covert packets at a low rate ($\lambda_B-\lambda$) as well as overt packets at rate $\lambda$. To do this, Alice performs a virtual Bernoulli splitting (p-thinning) on $B$, i.e., each time she wants to send a packet, she decides based on a Bernoulli random variable whether to send an overt or covert packet. Assuming that Alice always has covert packets to send, the proposed scheme requires Alice to also have overt packets always available so that if the result of the Bernoulli process leads sending an overt packet, she has one available to send. This suggests that Alice must first build up some number of overt packets in her buffer prior to starting the above procedure. 

In particular, Alice will employ a two-phase system. In the first phase, she will (slightly) slow down the transmission of packets from Jack so as to build up a backlog of packets in her buffer. In the next phase, she generates a renewal process with a rate higher than Jack's transmission rate, and starts sending overt and covert packets according to a Bernoulli splitting procedure as described above. 
To see how many packets Alice can buffer in the first phase, consider the following Lemma.

\begin{lem}\label{lem:2} If Alice can buffer packets on the link from Jack to Steve where the pdf of the inter-arrival times are $p_0(x)$, she can covertly buffer $\mathcal{O}\left(\sqrt{N}\right)$ packets in a packet stream of length $N$ as long as $p_1(x,\rho)=(1-\rho) p_0\left(\left(1-\rho\right)x\right)$ satisfies the following regulatory conditions \cite[Ch. 2.6]{kullback1968information}:
 \begin{align}
\label{c1}\text{• }&\frac{\partial \log{p_1}}{ \partial \rho}, \frac{\partial^2 \log{p_1}}{\partial \rho^2},  \frac{\partial^3 \log{p_1}}{\partial \rho^3} \text{ exist, } \forall \rho\in (0,1)\\
\text{• }\nonumber & \forall \rho \in (0,1),  \bigg|\frac{\partial p_1}{\partial \rho}\bigg| < F(x), \text{ s.t. } \int_{x=0}^{\infty}F(x)dx<\infty, \\
&\nonumber \bigg|\frac{\partial^2 p_1}{\partial \rho^2}\bigg| < G(x),\text{ s.t. } \int_{x=0}^{\infty}G(x)dx<\infty \\
&\nonumber  \bigg|\frac{\partial^3 \log p_1}{\partial \rho^3}\bigg| < H(x), \text{ s.t. } \int_{x=0}^{\infty}p_0(x)H(x)dx<\xi<\infty\\
\label{c2}  & \text{ where }\xi \text{ is independent of }\rho\\
\label{c3}\text{• } & \int_{x=0}^{\infty}\frac{\partial p_1(x,\rho)}{\partial \rho}\bigg|_{\rho=0}dx= \int_{x=0}^{\infty}\frac{\partial^2 p_1(x,\rho)}{\partial \rho^2}\bigg|_{\rho=0}dx=0
\end{align}
\noindent Conversely, if Alice buffers $\omega\left(\sqrt{N}\right)$ packets in a packet stream off length $N$, there exists a detector that Willie can use to detect such a buffering with arbitrarily low sum of error probabilities $\mathbb{P}_{FA} + \mathbb{P}_{MD}$.
\end{lem}

\begin{proof}
{\it (Achievability)} 

\textbf{Construction:} For a fixed number of packets $N$, Alice scales up the inter-arrival times of the packets by $\frac{1}{1-\rho}$ where $0<\rho<1$, i.e, if she receives the $i^{th}$ packet at $\tau_i$, she sends it at time $\frac{\tau_i}{1-\rho}$, as shown in Fig.~\ref{fig:Streched}. This allows her to transmit at a rate (slightly) lower than the rate she receives packets from Jack and therefore buffer packets. 

First we show that Alice can buffer $\mathcal{O}\left(\sqrt{N}\right)$ packets, and then we demonstrate the covertness.

\textit{Analysis}: ({\em Number of Buffered Packets}) Assume Alice sets 
\begin{align}
\label{rhoVal} \rho&= \frac{\epsilon}{\sqrt{c N}}
\end{align}
\noindent where $0<\epsilon<1$ and $c>0$ is a constant defined later. Then, she delays each packet by $\frac{1}{1-\rho}$ until time $\tau_N$ when she receives the $N^{th}$ packet. Since Alice sends the $i^{th}$ packet at $\frac{\tau_i}{1-\rho}$, we can observe that Alice sends the $i^{th}$ packet if and only if $\frac{\tau_i}{1-\rho}\leq \tau_N$. Therefore, the total number of packets that Alice transmits is $X\left(\tau_N \left(1-\rho\right)\right)$ and the total number of packets that Alice buffers is 
 \begin{figure}
 \begin{center}
\includegraphics[width=\textwidth/2,height=\textheight,keepaspectratio]{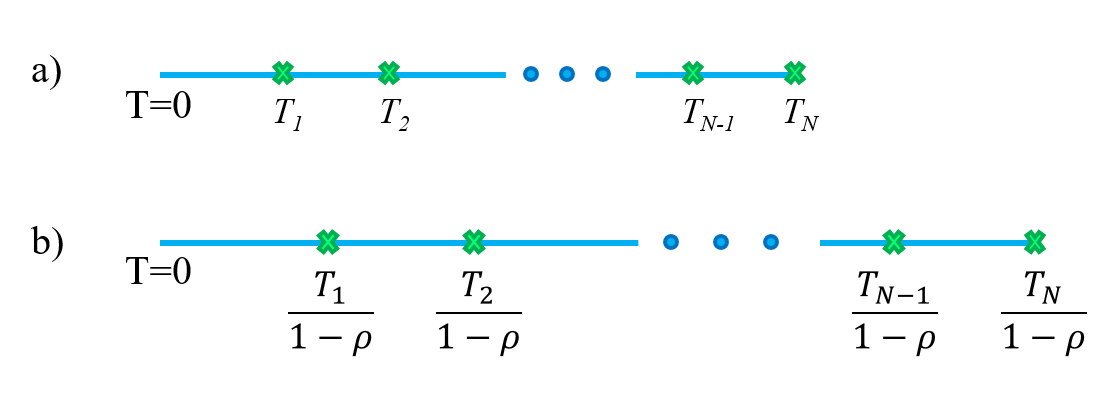}
 \end{center}
  \caption{a) Alice's received process b) The stretched version of Alice's received process when Alice uses a factor $\frac{1}{1-\rho}$.}
  \label{fig:Streched}
  \end{figure}
\begin{align}
\label{eq:0} N_c= N-X\left(\tau_N\left(1-\rho\right)\right)
\end{align}
\noindent We can show that (derived in the Appendix)
\begin{align}
\label{eq:5} \lim_{N \to \infty} \mathbb{P}\left(N_c \geq \epsilon \sqrt{\frac{N}{ 4c}}\right) = 1
\end{align}
\noindent Therefore, Alice can collect $\mathcal{O}\left(\sqrt{N}\right)$ packets in a packet stream of length $N$. 

 ({\em Covertness}) Now, we show that Alice's buffering is covert. We assume Willie knows the total number of packets that Alice has possibly collected, $N_c$, and the scaling factor that Alice has used for such a collection, $1-\rho$. Upon observing the first $N-N_c$ packets, Willie decides whether Alice has not done anything over the channel ($H_0$), or she has slowed down $N$ packets to buffer $N_c$ packets ($H_1$). If he applies an optimal hypothesis test that minimizes $\mathbb{P}_{FA}+\mathbb{P}_{MD}$ on the inter-arrival times, then \cite{bash13squarerootjsac}
 \begin{align} \label{eq:th10e00011} \mathbb{P}_{FA}+\mathbb{P}_{MD} \geq 1-     \sqrt{\frac{1}{2}  \mathcal{D}(\mathbb{P}_0 || \mathbb{P}_1)} \; 
 \end{align}
 \noindent where $\mathbb{P}_0$ and $\mathbb{P}_1$ are joint pdfs of the inter-arrival times when $H_0$ and $H_1$ are true respectively. Next, we show how Alice can lower bound the sum of average error probabilities by upper bounding $ \sqrt{\frac{1}{2} \mathcal{D}(\mathbb{P}_0 || \mathbb{P}_1)} $. Since inter-arrival times are i.i.d, 
 \begin{align}
 \mathbb{P}_0&=\prod_{i=1}^{N-N_c-1} p_0(x_i)\\
 \mathbb{P}_1&=\prod_{i=1}^{N-N_c-1} p_1(x_i,\rho)
 \end{align}
 \noindent where $p_0(x)$ and $p_1(x,\rho)$ are pdfs of a single inter-arrival time under $H_0$ and $H_1$ respectively. Therefore, 
 \begin{align}
 \label{lem2:1}  \mathcal{D}(\mathbb{P}_0 || \mathbb{P}_1) = (N-N_c-1)   \mathcal{D}\left({p}_0(x) || {p}_1(x,\rho)\right)
 \end{align}
 \noindent Note that $p_1(x,\rho)=\left(1- \rho\right)p_0\left(\left(1- \rho\right) x\right)$ represents the family of pdfs that are scaled version of $p_0(x)$. Since the regulatory conditions (\ref{c1}-\ref{c3}) hold, \cite[Ch. 2.6]{kullback1968information}
 \begin{align}
   \label{dklkhaf}\mathcal{D}\left({p}_0(x) || {p}_1(x)\right) =  \frac{c\rho^2}{2} +  \mathcal{O}\left(\rho^3\right) \text{ as } \rho \to 0
 \end{align}
\noindent where the constant $c>0$ is (derived in the Appendix)
\begin{align}
\label{c4}  c=-1+ \int_{x=0}^{\infty}  p_0(x) {x^2} \left(\frac{d \log p_0(x)}{dx}\right)^2 dx
\end{align}
 \noindent Thus
 \begin{align}
\label{lem2:2} \mathcal{D}(\mathbb{P}_0 || \mathbb{P}_1) &= (N-N_c-1) \left(\frac{c \rho^2}{2} +  \mathcal{O}\left(\rho^3\right) \right) \leq N \left(\frac{c \rho^2}{2} +  \mathcal{O}\left(\rho^3\right) \right) \text{ as } \rho \to 0
 \end{align}
 \noindent Hence, by \eqref{rhoVal}
 \begin{align}
 \nonumber  \lim_{N\to\infty}\sqrt{\frac{1}{2} \mathcal{D}(\mathbb{P}_0 || \mathbb{P}_1)} &\leq\lim_{N\to\infty} \epsilon \sqrt{\frac{N}{2 N}} \leq \epsilon
 \end{align}
 \noindent Thus, by \eqref{eq:th10e00011}, $\mathbb{P}_{FA}+\mathbb{P}_{MD} {\geq} 1-\epsilon $ as $N\to \infty$ and Alice can covertly buffer $\mathcal{O}\left(\sqrt{N}\right)$ packets in a packet stream of length $N$.

{\it (Converse)} Suppose that Willie observes $N-N_c$ packets which have $N-N_c-1$ interval-arrival times and wishes to detect whether Alice has done nothing over the channel ($H_0$) or she has delayed each packet by $\frac{1}{1-\rho}$, where $0<\rho<1$ is random variable, and buffered $N_c$ packets ($H_1$). Note that when $H_0$ is true, the inter-arrival times are samples of $p_0(x)$ and when $H_1$ is true, the inter-arrival times are the samples of $p_1(x,\rho)=\left(1-\rho\right) p_0\left(\left(1-\rho\right) x\right)$. Since Willie knows $p_0(x)$, he knows the expected number of inter-arrival times. Therefore, he calculates the sum of $N-N_c-1$ inter-arrival times $S$ for the observed packets and performs a hypothesis test by setting a threshold $U$ and comparing $S$ with $(N-N_c-1) \lambda^{-1} +U$. If $S<(N-N_c-1)\lambda^{-1}+U$, Willie accepts $H_0$; otherwise, he accepts $H_1$. Let $N'=N-N_c-1$
\begin{align}
\label{eq:PFAupperbound12} \mathbb{P}_{FA} = \mathbb{P}\left(S>N'\lambda^{-1} + U | H_0 \right)=\mathbb{P}\left((S-N') \lambda^{-1} > U | H_0 \right)
&\leq\mathbb{P}\left(|S-N'| \lambda^{-1} > U | H_0 \right)
\end{align}
\noindent When $H_0$ is true, Willie observes a renewal process with rate $\lambda$ and inter-arrival variance of $\sigma^2$; hence,
\begin{align}
\label{eq:th1con12}\mathbb{E}\left[S\left|H_0\right.\right]&=N' \lambda^{-1}\\
\label{eq:th1con22}\mathrm{Var}\left[S\left|H_0\right.\right]&={N' \sigma^2}
\end{align}
\noindent Therefore, applying Chebyshev's inequality on~\eqref{eq:PFAupperbound12} yields $\mathbb{P}_{FA} \leq \frac{N' \lambda^{-1}}{U^2 }$. Therefore, if Willie sets 
\begin{align}
\label{con1u} U=\sqrt{\frac{N' }{ \lambda \alpha}}
\end{align}
\noindent for any $0<\alpha<1$, then
\begin{align}
\mathbb{P}_{FA}\leq \alpha
\end{align}
 Next, we will show that if Alice collects $N_c=\omega\left(\sqrt{N}\right)$ packets, she will be detected by Willie with high probability. When $H_1$ is true, since Willie observes a renewal point process in which the pdf of inter-arrival times are $p_1(x,\rho)=\left(1-\rho\right) p_0\left(\left(1-\rho\right) x\right)$,
\begin{align}
\label{eq:th1con32}\mathbb{E}\left[S\left|H_1\right.\right]&= \frac{N'}{\lambda(1-\rho)} \\
\label{eq:th1con42}\mathrm{Var}\left[S\left|H_1\right.\right]&= \frac{N)  \sigma^2}{\left(1-\rho\right)^2}
\end{align}
\noindent Now, consider $\mathbb{P}_{MD}$. We can show that (derived in the Appendix)
\begin{align} 
\label{lem2pmd} \mathbb{P}_{MD}&\leq  \frac{ \sigma^2 \lambda^2}{  N' {\rho}^2}
\end{align}
\noindent Since Willie knows $\rho = \omega\left(\frac{1}{\sqrt{N}}\right)$,
\begin{align} 
\lim\limits_{N \to \infty} \mathbb{P}_{MD} =0 
\end{align}
\noindent Therefore, Willie can achieve $\mathbb{P}_{MD} < \beta$ for any $0<\beta<1$. Combined with the results for the probability of false alarm above, if Alice collects $N_c=\omega\left({\sqrt{N}}\right)$, Willie can choose a $U=\sqrt{\frac{N'}{\lambda \alpha}}$ to achieve any (small) $\alpha>0$ and $\beta>0$ desired.
\end{proof}

Next, we present and prove the results for Scenario 1 using the results of Lemma \ref{lem:2}.
\begin{thm} \label{th:3} Consider Scenario 1 with conditions (\ref{c1}-\ref{c3}) true. Then, Alice can covertly insert $\mathcal{O}(\sqrt{N})$ packets in a packet stream of length $N$. 
\end{thm}

\begin{proof}
{\it (Achievability)} 

\textbf{Construction:} For a fixed number of packets $N$, Alice's transmission includes two phases: a buffering phase and a transmission phase. During the buffering phase of length $ \tau_{ \psi N}$, where $0<\psi<1$ is a parameter to be defined later, Alice scales the inter-arrival times of the first $\psi N$ packets of Jack's transmitted stream to build up packets in her buffer. Based on the results of Lemma~\ref{lem:2}, she buffers $\mathcal{O}\left(\sqrt{N }\right)$ packets. In the second phase of length $\tau_N-\tau_{\psi N}$, Alice scales $p_0(x)$ to 
\begin{align}
\label{p1th3} p_1(x,\rho)=\frac{p_0\left({x}/\left(1-\rho\right)\right)}{1-\rho}
\end{align}
\noindent where $0<\rho<1$, then generates inter-arrival times according to $p_1(x,\rho)$ that represents a renewal process entitled ``Overt-Covert Process''with a rate $\lambda_{oc}$ which is higher than Jack's transmission rate. According to this Overt-Covert Process, Alice starts sending both overt and covert packets. But, to decide when she should send a covert or overt packet, she uses a Bernoulli splitting (p-thinning) procedure, i.e., each time she wants to send a packet, first she generates a random variable according to a Bernoulli distribution with 
\begin{align}
\label{bernulli1}\mathbb{P}\left(\text{Success}\right)=\rho
\end{align}
\noindent If she observes ``Success'', she sends a covert packet, otherwise, she sends an overt packet. 

\textit{Analysis} ({\em Number of Packets}) The rate at which Alice transmits packets in the Overt-Covert Process is 
\begin{align}
\label{rate:1} \lambda_{oc} &= \left(\int_{x=0}^{\infty}x p_1(x,\rho)dx\right)^{-1}=  \left(\int_{x=0}^{\infty} \frac{x p_0\left(x /\left(1-\rho\right) \right)}{1-\rho}dx\right)^{-1}\\
\label{rate:2} &=  \left((1-\rho)\int_{x=0}^{\infty} {x p_0\left(x  \right)}dx\right)^{-1}\\
\label{rate:3} &=  (1-\rho)^{-1}\lambda
\end{align}
\noindent where~\eqref{rate:2} follows from~\eqref{rate:1} by change of variable. Denote the total number of overt and covert packets that Alice transmits in the second phase by $N_{oc}$. Since Alice sends a stream of overt and covet packets in which the locations of covert packets are chosen according to Bernoulli random variables, the total number of covert packets that Alice inserts is
\begin{align}
\label{Nc} N_c = \sum\limits_{i=1}^{N_{oc}} b_i
\end{align}
\noindent where each $b_i$ is a Bernoulli random variable with 
\begin{align}
\label{pbi} \mathbb{P}(b_i=1)=\rho
\end{align}
\noindent Similar to the arguments that leads to~\eqref{eq:5} we can show that
\begin{align}
\label{NocN1} \lim_{N \to \infty}\mathbb{P}\left(N_{oc} \geq \frac{N \left(1-\psi\right)}{2}\right) &=1\\
\label{Ncasymp}\lim_{N \to \infty}\mathbb{P}\left(N_{c} \geq \frac{\rho N_{oc}}{2}\right)
&=1
\end{align}
\noindent Note that for any two events $\mathcal{E}_1$ and $\mathcal{E}_2$, if $\mathbb{P}\left(\mathcal{E}_1\right)=1$ and $\mathbb{P}\left(\mathcal{E}_2\right)=1$, then $\mathbb{P} (\overline{\mathcal{E}_1} \cup \overline{\mathcal{E}_2}) \leq  \mathbb{P} (\overline{\mathcal{E}_1} )  + \mathbb{P} ( \overline{\mathcal{E}_2}) =0$ and therefore, $\mathbb{P} (\overline{\mathcal{E}_1} \cup \overline{\mathcal{E}_2}) =0$ and $\mathbb{P} ({\mathcal{E}_1} \cap {\mathcal{E}_2}) =1$. Now, if we let $\mathcal{E}_1 = \{N_{oc} \geq \frac{N  \left(1-\psi\right)}{2\left(1-\rho\right)}\}$ and $\mathcal{E}_2 = \{N_{c} \geq \frac{\rho N_{oc}}{2}\}$, then
\begin{align}
\nonumber \lim_{N \to \infty}\mathbb{P}\left(N_{c} \geq \frac{\rho N \left(1-\psi\right)}{4}\right) &=1
\end{align}
\noindent Now, if Alice sets
\begin{align}
\label{th3rhoVal} \rho&=  \frac{\epsilon}{\sqrt{{2 c N (1-\psi)}}}
\end{align}
\noindent then
\begin{align}
\label{th3N_c} \lim_{N \to \infty}\mathbb{P}\left(N_{c} \geq \frac{ \epsilon }{4} \sqrt{\frac{N (1-\psi)}{2c}}\right) &=1
\end{align}
\noindent Thus, Alice can insert $\mathcal{O}\left(\sqrt{N}\right)$ packets in a packet stream of length $N$. 

({\em Covertness}) Assume Willie knows Alice's transmission scheme and parameters as well as the time she starts and ends the first and second phase. He also knows the number of covert packets that Alice has possibly inserted into the channel, $N_c$. In the first phase, she receives a packet stream of length $\psi N$. Therefore, by the results of Lemma~\ref{lem:2}, Alice buffers $m=\mathcal{O}\left(\sqrt{N}\right)$ packets where
\begin{align}
\label{th3m} \lim\limits_{N \to \infty} \mathbb{P}\left(m \geq \sqrt{\frac{N \psi}{4c}}\right)=1 
\end{align}
packets while lower bounding Willie's sum of error probabilities $\mathbb{P}_{FA}+\mathbb{P}_{MD}$ by $1-\epsilon$. Therefore, her buffering is covert. 

In the second phase, Alice inserts $N_c$ covert packets in a packet stream of $N_{oc}$ overt and covert packets. Willie, upon observing inter-arrival times of $N_{oc}$ packets, decides whether Alice has not done anything over the channel and therefore the inter-arrival times of the packets are governed by pdf $p_0(x)$, ($H_0$), or she has inserted $N_c$ covert packets along with $N_{oc}-N_c$ overt packets and therefore the inter-arrival times of the packets are governed by $p_1(x,\rho)=\frac{p_0(x/(1-\rho))}{1-\rho}$, ($H_1$). If Willie applies an optimal hypothesis test that minimizes $\mathbb{P}_{FA}+\mathbb{P}_{MD}$ on the inter-arrival times, then \cite{bash13squarerootjsac}
\begin{align} \label{eq:th30e000111} \mathbb{P}_{FA}+\mathbb{P}_{MD} \geq 1-     \sqrt{\frac{1}{2}  \mathcal{D}\left(\mathbb{P}_0 || \mathbb{P}_1\right)} \; 
\end{align}
\noindent where $\mathbb{P}_0$ and $\mathbb{P}_1$ are joint pdfs of the inter-arrival times when $H_0$ and $H_1$ are true respectively. Next, we show how Alice can lower bound the sum of average error probabilities by upper bounding $ \sqrt{\frac{1}{2} \mathcal{D}(\mathbb{P}_0 || \mathbb{P}_1)} $. Since inter-arrival times are i.i.d, 
 \begin{align}
 \mathbb{P}_0&=\prod_{i=1}^{N_{oc}-1} p_0(x_i)\\
 \mathbb{P}_1&=\prod_{i=1}^{N_{oc}-1} p_1(x_i,\rho)
 \end{align}
 \noindent Thus,
 \begin{align}
\mathcal{D}(\mathbb{P}_0 || \mathbb{P}_1) = (N_{oc}-1)   \mathcal{D}\left({p}_0(x) || p_1(x,\rho)\right)
\label{th3:1} & \leq  N_{oc}   \mathcal{D}\left({p}_0(x) || p_1(x,\rho)\right)
 \end{align}
\noindent We can easily see that when the conditions (\ref{c1}-\ref{c3}) hold for $p_1(x,\rho) = (1-\rho) p(((1-\rho))x)$ in Lemma~\ref{lem:2}, they hold for $p_1(x,\rho)=\frac{p_0\left(\left(1-\rho\right)x\right)}{1-\rho}$ as well. Therefore~\cite[Ch. 2.6]{kullback1968information}
\begin{align}
\label{dklkhaf2}\mathcal{D}\left({p}_0(x) || {p}_1(x)\right) &= \frac{c \rho^2}{2} +  \mathcal{O}\left(\rho^3\right) \text{ as } \rho \to 0
\end{align}
\noindent Note that the constant $c$ is the same as the one in~\eqref{c4} (the proof follows the lines of proof of~\eqref{c4} in the Appendix with minor modifications). Thus, we can show that (proved in the Appendix)
\begin{align}
\label{th3e6}\lim\limits_{N \to \infty} \mathbb{E} \left[\sqrt{\frac{\mathcal{D}(\mathbb{P}_0 || \mathbb{P}_1)}{2}}\right] &< \epsilon
\end{align}
  \noindent Consequently, by \eqref{eq:th30e000111}, $\mathbb{E}[\mathbb{P}_{FA}+\mathbb{P}_{MD}] > 1-\epsilon $ as $N\to \infty$. Combined with the results of the covertness in the first phase, Alice can covertly insert $\mathcal{O}\left(\sqrt{N}\right)$ packets in a packet stream of length $N$.

({\em Failure Analysis}) In the second phase, Alice avoids a ``failure'' event, in which she cannot send an overt packet from her buffer because she has run out of packets. Next, we show that Alice can choose $\psi$ such that she achieves $\mathbb{P}_f< \zeta$ for any $ \zeta> 0$, where $\mathbb{P}_f$ is the probability of the event ``failure''. 

Since Bob removes all of the covert packets from the channel and transmits only overt packets to Steve, Alice's overt packet transmission rate is the same as Bob's transmission which is
\begin{align}
  \lambda_o &= \lambda_{oc} \left(1-\mathbb{P}\left(\text{Success}\right)\right)
\label{lamovert}  =\frac{\lambda}{ \left(1-\rho\right)} \left(1-\rho\right) = \lambda
\end{align}
\noindent Thus, Alice's transmits overt packets at the same rate she receives overt packets from Jack. Therefore, similar to Scenario 2, we can analyze the ``failure'' event by modeling the receiving and transmitting of overt packets by a random walk problem which has the same probability of moving from location $z$ to $z+1$ or $z-1$. In Scenario 2, because the timing of the received and transmitted packets are modeled by a Poisson point process, which has the memoryless property, the random walk has the equal probability of moving $z+1$ or $z-1$, and these probabilities does not depend on the $z$. However, in Theorem~\ref{th:3}, this property does not hold, and thus the random walk is not a regular random walk. Here, we approximate this random walk with a regular random walk~\footnote{The accurate analysis will be provided in future versions}. 

Note that $N - {N \psi}= N (1-\psi)$ is the number of packets that Alice receives in the second phase and let the random variable 
\begin{align}
\label{valNo} N_o=N_{oc}-N_c
\end{align}
\noindent be the total number of overt packets that Alice transmits in the second phase. Therefore, 
\begin{align}
\label{valk} K=N(1-\psi)+N_s
\end{align}
\noindent is the total number of received and transmitted overt packets during the second phase. By the law of total probability we can show that
\begin{align}
\label{eq:th2031}  \mathbb{P}_f \leq   \mathbb{P}\left(\mathcal{F} | \mathcal{E}_1 \cap \mathcal{E}_2 \right) +\mathbb{P}\left( \overline{\mathcal{E}_1 
}\right) + \mathbb{P}\left( \overline{\mathcal{E}_2 
}\right)
\end{align}
\noindent where $\mathcal{F}$ is the ``failure'' event, $\mathcal{E}_1 = \Big\{m \geq \epsilon \sqrt{ \frac{N\psi}{4c}}\Big\}$, $\mathcal{E}_2 = \{K\geq 4 N (1-\psi)\}$, and the total number of received and transmitted overt packets packets during the second phase is
\begin{align}
\label{th3K} K=N(1-\psi)+N_o
\end{align}
\noindent By~\eqref{th3m}. 
\begin{align}
\label{eq:th2041} \lim_{N \to \infty}\mathbb{P}\left(\overline{\mathcal{E}_1}\right)=0
\end{align}
Also, we can easily show that (derived in the Appendix)
\begin{align}
\label{randw1} \lim_{N \to \infty}\mathbb{P}\left(\overline{\mathcal{E}_2}\right) = \lim\limits_{N \to \infty}  \mathbb{P}\left(K\geq 4N\left(1-\psi\right)\right) = 0
\end{align}
Now, consider $ \mathbb{P}\left(\mathcal{F}| \mathcal{E}_1 \cap\mathcal{E}_2\right)$. Similar to the arguments that leads to\cite[Eq. 27]{soltani2015covert}, we can show that if $m'= \epsilon \sqrt{ \frac{N\psi}{4c}} $, $k'= 4 N \left(1-\psi\right)$, then
\begin{align}
  \lim\limits_{N \to \infty} \mathbb{P}\left(\mathcal{F}| \mathcal{E}_1 \cap\mathcal{E}_2\right) \nonumber \leq 1- \lim\limits_{k'\to \infty} \mathbb{P}\left(\mathcal{F}| \mathcal{E}_1 \cap\mathcal{E}_2\right)
\leq 1-\operatorname{erf}\left(\frac{\epsilon \sqrt{  \psi N} }{ \sqrt{32 N c \left(1-\psi\right)}}\right)
\leq  1-\operatorname{erf}\left({\epsilon  }\sqrt{\frac{\psi}{32 c\left(1- \psi\right)}}\right)
\end{align}
\noindent Therefore, if $\psi$ is chosen such that: 
\begin{align}
\nonumber \frac{ \psi}{1- \psi} =\left( \frac{\sqrt{32 c} }{\epsilon} \operatorname{erf}^{-1} \left(1-{\zeta}\right)\right)^2 
\end{align}
\noindent then, 
\begin{align}
\label{eq:th3.r2}  \lim\limits_{N\to\infty} \mathbb{P}\left(\mathcal{F}| \mathcal{E}_1 \cap\mathcal{E}_2\right) \leq \zeta
\end{align}
\noindent Therefore, by~\eqref{eq:th2031},~\eqref{eq:th2041},~\eqref{randw1}, and~\eqref{eq:th3.r2}, Alice can achieve $\mathbb{P}_f < \zeta$ for any $0<\zeta<1$. 

In the calculation of number of inserted packets as well as the covertness analysis, we have proved the results given ``failure'' does not occur. Since $\mathbb{P}_f < \zeta$ for any $0<\zeta<1$, the results hold for the general case where the ``failure'' may occur.

Next, we provide the converse for the proposed theorem, given Alice uses the two phased scenario discussed above. The generalization of that will be provided with the final submission of the paper.
\end{proof}
\begin{thm}
Conversely, if Alice attempts to insert $\omega\left(\sqrt{N}\right)$ packets in a packet stream of length $N$, there exists a detector that Willie can use to detect her with arbitrarily low sum of error probabilities $\mathbb{P}_{FA} + \mathbb{P}_{MD}$.
\end{thm}
\begin{proof}
Suppose Willie knows the total number of packets transmitted by Jack $N$, and the fact that Alice's communication includes two phases as described above. Also, he knows $\psi$, i.e, when the each of the phases starts and end. Willie, wishes to decides whether Alice has done nothing over the channel ($H_0$) or she has performed a two phased scheme on the channel ($H_1$) to transmit covert packets. Also, Willie knows that Alice in which in the first phase she has slowed down the first $N \psi$ packets to buffer some packets, where $0<\psi<1$, and in second phase she has generated a renewal process in which the inter-arrival times are modeled by $p_1(x,\rho)=\frac{1}{1-\rho} p_0\left(\frac{x}{1-\rho}\right)$. Then, according to the renewal process she has transmitted covert $N_c$ packets along with $N-N_c$ overt packets. We assume Willie knows $\psi$ but $0<\rho<1$ is a random variable.

To decide about Alice's communication, Willie disregards the first $N \psi$ packets and only considers the other $N''+1=N (1-\psi)$ packets which have $N''$ interval-arrival times. Since he knows when $H_0$ is true, the inter-arrival times are samples of $p_0(x)$, he knows the expected number of inter-arrival times. Therefore, he calculates the sum of inter-arrival times $S_A$ for the $N''$ selected packets and performs a hypothesis test by setting a threshold $U_A$ and comparing $S_A$ with $N''\lambda^{-1} +U_A$. If $S<N''\lambda^{-1}+U_A$, Willie accepts $H_0$; otherwise, he accepts $H_1$. Consider $\mathbb{P}_{FA}$
\begin{align}
 \mathbb{P}_{FA} = \mathbb{P}\left(S_A>N''\lambda^{-1} + U_A | H_0 \right)  =\mathbb{P}\left(S_A-N'' \lambda^{-1} > U_A | H_0 \right)
\label{eq3:PFAupperbound12}&\leq\mathbb{P}\left(|S_A-N'' \lambda^{-1}| > U_A | H_0 \right)
\end{align}
\noindent When $H_0$ is true, Willie observes a renewal process with rate $\lambda$ and inter-arrival variance of $\sigma^2$; hence,
\begin{align}
\label{eq3:th1con12}\mathbb{E}\left[S_A\left|H_0\right.\right]&=N'' \lambda^{-1}\\
\label{eq3:th1con22}\mathrm{Var}\left[S_A\left|H_0\right.\right]&={N'' \sigma^2}
\end{align}
\noindent Therefore, applying Chebyshev's inequality on~\eqref{eq3:PFAupperbound12} yields $\mathbb{P}_{FA} \leq \frac{N'' \lambda^{-1}}{U_A^2 }$. Therefore, if Willie sets 
\begin{align}
\label{th3con1u} U_A=\sqrt{\frac{N''}{ \lambda \alpha}}
\end{align}
\noindent for any $0<\alpha<1$, then
\begin{align}
\mathbb{P}_{FA}\leq \alpha
\end{align}
 Next, we will show that if Alice inserts $N_c=\omega\left(\sqrt{N}\right)$ packets, she will be detected by Willie with high probability. When $H_1$ is true, since Willie observes a renewal point process in which the pdf of inter-arrival times are $p_1(x,\rho)=\frac{1}{1-\rho} p_0\left(\frac{x}{1-\rho}\right)$,
\begin{align}
\label{eq:th3con32}\mathbb{E}\left[S_A\left|H_1\right.\right]&= \frac{N''(1-\rho)}{\lambda} \\
\label{eq:th3con42}\mathrm{Var}\left[S_A\left|H_1\right.\right]&= {N''(1-\rho)^2 \sigma^2}
\end{align}
\noindent Now, consider $\mathbb{P}_{MD}$. Similar to the arguments that leads to~\eqref{lem2pmd}, we can show that 
\begin{align} 
\label{th3pmd} \mathbb{P}_{MD}&\leq  \frac{ \sigma^2 \lambda^2 \left(1-\rho\right)^2}{N'' {\rho}^2}
\end{align}
\noindent Also, if Alice inserts $N_c=\omega(\sqrt{N})$ packets, then it must be that $\rho= \omega\left(\frac{1}{\sqrt{N}}\right)$. Therefore,
\begin{align} 
\lim\limits_{N \to \infty} \mathbb{P}_{MD}=0
\end{align}
\noindent Therefore, Willie can achieve $\mathbb{P}_{MD} < \beta$ for any $0<\beta<1$. Combined with the results for probability of false alarm above, if Alice inserts $N_c=\omega\left({\sqrt{N}}\right)$, Willie can choose a $U_A=\sqrt{\frac{N''}{\lambda \alpha}}$ to achieve any (small) $\alpha>0$ and $\beta>0$ desired.
\end{proof}
In this scenario, we saw that Alice is allowed to buffer packets transmitted by Jack and release them when it is necessary; thus she is able to alter the timings of the packets. This suggests that Alice can also alter the timings of the packets to send information to Bob (as in Scenario 2) to achieve a higher throughput for sending covert information. However, this would require Alice and Bob to share a secret key (unknown to adversary Willie) prior to the communication which is not possible in many scenarios. Also, packet insertion works over channels for which sending the information through packet timings does not work, such as complicated channels (e.g. mixed with other flows, then separated) which change the timings of the packets significantly and channels with zero capacity when packet timing approaches are employed (e.g. deterministic queues). If we assume Alice and Bob can share a codebook and altering of timings in the channel can be modeled by a queue, we can consider sending information via packet timing, which is discussed in the next scenario.
\subsection{General Renewal Model, Packet Timing (Scenario 2)}
In this section, we consider Scenario 2: In a Non-Poisson channel, Willie can authenticate packets to determine whether or not they were generated by the legitimate transmitter Jack. Therefore, Alice cannot insert packets into the channel; rather, we assume that Alice is able to buffer packets and release them when she desires; hence, she can encode information in the inter-packet delays by using a secret codebook shared with Bob. 

Here, similar to the Poisson case, each of Alice's codewords will consist of a sequence of inter-packet delays to be employed to convey the corresponding message. Also, Alice will employ a two-phase system. In the first phase, she will (slightly) slow down the transmission of packets from Jack to Steve so as to build up a backlog of packets in her buffer. Then, during the codeword transmission phase, she will release packets from her buffer with the inter-packet delays prescribed by the codeword corresponding to the message, while continuing to buffer arriving packets from Jack. To see how much Alice can slow down the packet stream from Jack to Steve without it being detected by warden Willie, we use the results of the Lemma \ref{lem:2}. Next, we propose the capacity of $G/M/1$ in Lemma \ref{lem:3}. Then, we calculate the number of packets that Alice should accumulate in her buffer by the start of the second phase so as to, with high probability, have a packet in her buffer at all of the times required by the codeword. Finally, we consider the throughput of Alice's communication in Theorem \ref{th4}.

Consider the $G/M/1$ queue defined in Section \ref{sec:1}. We propose and prove the upper bound on the its capacity in the following Lemma.

\begin{lem} \label{lem:3}
The $G/M/1$ queue with service rate $\mu$ and input rate $\lambda<\mu$ satisfies
\begin{align} \label{q:00}
C(\lambda) \geq \lambda \log{\frac{\mu}{\lambda}}-\lambda  \mathcal{D}\left(p_0\left(x\right)||e_{\lambda}\left(x\right)\right)
\end{align}
\noindent where $e_{\lambda}(x)=\lambda e^{-\lambda x}$.
\end{lem}

\begin{proof}
{\it (Achievability)} 

\textbf{Construction:} We assume that at the time of transmission, the queue is in equilibrium. We start the transmission by sending the first packet, and we consider the time that this packet arrives at the queue is time zero. Then, using the shared codebook between the transmitter and the receiver, the transmitter encodes the message into $n$ inter-packet delays $A^n=\left(A_1,\cdots,A_n\right)$; i.e, $A_i$ is the time elapsed between $i^{th}$ and $(i+1)^{th}$ packet for $1\leq i \leq n$. We denote the time that the first packet (called packet zero in \cite{verdubitsq}) spends in the queue by $D_0$, and let $D^n=\left(D_0,D_1,\cdots,D_n\right)$, where $D_i$ is the inter-departure time between the $i^{th}$ and $(i+1)^{th}$ packet for $1\leq i \leq n$.

For the $G/M/1$ queue, we denote the joint pdf of the inter-arrival times by $\mathbb{Q}_{A^n}\left(a^n\right)$, joint pdf of the vector $D^n$ by $\mathbb{Q}_{D^n}\left(d^n\right)$, joint pdf of $A^n$ and $D^n$ by $\mathbb{Q}_{A^n,D^n}\left(a^n,d^n\right)$, conditional pdf of $A^n$ given $D^n$ by $\mathbb{Q}_{A^n|D^n}\left(a^n|d^n\right)$ and conditional pdf of $D^n$ given $A^n$ by $\mathbb{Q}_{D^n|A^n}\left(d^n|a^n\right)$.

For a special case in which the inter-arrival times are modeled by exponential random variables with exponential pdf $e_{\lambda}\left(x\right)$, the $G/M/1$ queue is a $M/M/1$ queue. For this queue, we denote all of the above joint pdfs by the letter $\mathbb{P}$ instead of $\mathbb{Q}$.
 
\textit{Analysis}: We can obtain the capacity of the queue by \cite{verdu1994general}:
\begin{align}
\label{q:cap}C(\lambda)=\lambda \sup\limits_{A^n}\,  \mathbb{\underline{I}}\left(A^n;D^n\right) 
\end{align}
\noindent where \cite{han1993approximation} 
\begin{align}
\label{q:cap2}\mathbb{\underline{I}}\left(A^n;D^n\right)=\sup\bigg\{\alpha\in \mathbb{R}  :\mathbb{P} \left(\frac{1}{n} i_{A^n;D^n}\left(a^n;d^n\right) \leq \alpha \right)\bigg\}
\end{align}
\noindent is the \textit{liminf in probability} of the sequence of normalized information densities
\begin{align}
\frac{1}{n} i_{A^n;D^n}\left(a^n;d^n\right) = \frac{1}{n}\log \frac{\mathbb{Q}_{D^n|A^n}\left(d^n|a^n\right)}{\mathbb{Q}_{D^n}\left(d^n\right)}
\end{align}
Comparing \eqref{q:cap} with the formula for capacity in \cite{verdu1994general}, we see an extra $\lambda$ in the right hand side (RHS) of \eqref{q:cap} that is due to differences between the definitions of capacity in \cite{verdu1994general} and here.
To show that \eqref{q:00} is true, by \eqref{q:cap},\eqref{q:cap2}, it is enough to show that there exists a sequence of random variables $A_1,A_2,\cdots$ such that
\begin{align}
\sup\bigg\{\alpha\in \mathbb{R}  :\mathbb{P} \left(\frac{1}{n} i_{A^n;D^n}\left(a^n;d^n\right) \leq \alpha \right)\bigg\}\geq
\label{eq:th4.1} &\log{\frac{\mu}{\lambda}}-  \mathcal{D}\left(\mathbb{Q}_A\left(x\right)||e_{\lambda}\left(x\right)\right)
\end{align}
\noindent To establish~\eqref{eq:th4.1}, it is sufficient to show that there exists a sequence of random variables $A_1,A_2,\cdots$ such that for every $\gamma>0$
\begin{align}
\label{q:0} \lim\limits_{n\to \infty}\mathbb{P}\left[\frac{1}{n} \log \frac{\mathbb{Q}_{D^n|A^n}\left(d^n|a^n\right)}{\mathbb{Q}_{D^n}\left(d^n\right)}    \log{\frac{\mu}{\lambda}}- \mathcal{D}\left(\mathbb{Q}_A(x)||e_{\lambda}(x)\right)-\gamma\right] = 0
\end{align}
We can easily prove that (derived in the Appendix)
\begin{align}
\frac{1}{n}\log \frac{\mathbb{Q}_{D^n|A^n}\left(d^n|a^n\right)}{\mathbb{Q}_{D^n}\left(d^n\right)}&= \frac{1}{n} \log{\frac{\mathbb{P}_{D^n|A^n}\left(d^n|a^n\right)}{\mathbb{P}_{D^n}\left(d^n\right)}}+\frac{1}{n} \log{\frac{\mathbb{Q}_{A^n|D^n}\left(a^n|d^n\right)}{\mathbb{P}_{A^n|D^n}\left(a^n|d^n\right)}
}
 \label{q:10}\phantom{=}+\frac{1}{n}\log{\frac{\mathbb{P}_{A^n}\left(a^n\right)}{\mathbb{Q}_{A^n}\left(a^n\right)}}
\end{align}
 Note that in the above equation, the pdfs denoted by letter $\mathbb{P}$ are related to $M/M/1$ queue, but the arguments in the above equation are the random variables related to the $G/M/1$ queue. 

Consider the three terms on the right hand side of \eqref{q:10}. We can show that for all $\gamma>0$ (proved in the Appendix)
\begin{align}
&\label{q:6}  \lim\limits_{n \to \infty}\mathbb{P}\left({\frac{1}{n} \log{\frac{\mathbb{P}_{D^n|A^n}\left(d^n|a^n\right)}{\mathbb{P}_{D^n}\left(d^n\right)}} < \log{\frac{\mu}{\lambda}} -\gamma/3 }\right)=0\\
&\label{q:7} \lim\limits_{n \to \infty} \mathbb{P}\left(\frac{1}{n} \log{\frac{\mathbb{Q}_{A^n|D^n}\left(a^n|d^n\right)}{\mathbb{P}_{A^n|D^n}\left(a^n|d^n\right)}} < -\gamma/3\right)=0\\
 \label{q:8} &  \lim\limits_{n \to \infty}\mathbb{P} \left( \frac{1}{n}\log{\frac{\mathbb{P}_{A^n}\left(a^n\right)}{\mathbb{Q}_{A^n}\left(a^n\right)}} 
 <- \mathcal{D}\left(\mathbb{Q}_{A}\left(x\right)||e_{\lambda}\left(x\right)\right) - \gamma/3 \right)=0
\end{align}
\noindent Therefore, \eqref{q:10}-\eqref{q:8} yield \eqref{q:0} and the proof is complete. 
\end{proof}

\begin{thm} \label{th4} Consider Scenario 4 with conditions (\ref{c1}-\ref{c3}) true and
\begin{align} \label{q:01}
\lambda \log{\frac{\mu}{\lambda}}\lambda  -\mathcal{D}\left(p_0\left(x\right)||e_{\lambda}\left(x\right)\right) > 0
\end{align}
\noindent where $e_\lambda(x)=\lambda e^{-\lambda x}$. By embedding information in the inter-packet delays, Alice can covertly and reliably transmit $\mathcal{O}\left(\lambda N \right)$ bits to Bob in a packet stream of length $N$.
\end{thm}
\begin{proof} {\it (Achievability)} 

\textbf{Construction:} To establish covert communication over the timing channel, Alice and Bob share a secret key (codebook) to which Willie does not have access. To build a codebook, a set of $M$ independently generated codewords $\{C(J_i)\}_{i=1}^{i=M}$ are generated for messages $\{J_i\}_{i=1}^{i=M}$ according to realizations of a renewal process with inter-arrival pdf $p_0(x)$ that mimics the overt traffic on the channel between Jack and Steve, where $M$ is the size of the codebook. In particular, to generate a codeword $C(J_i)$, the renewal process of the packets transmitted by Jack is simulated, i.e., $C(J_i)$ consists of inter-arrival times $A_1, \cdots,A_{N(1-\psi)}$ that generated according to the pdf $p_0(x)$. For each message transmission, Alice uses a new codebook to encode the message into a codeword. According to the codebook, each message corresponds to a codeword that is a series of inter-packet delays. Alice starts the transmission of the codeword by sending the first packet and then applies the inter-packet delays to the packets that are being transmitted from Jack to Steve (see Fig. \ref{fig:codebook}). On the other hand, Bob knows when to start reading the inter-packet delays and decode them based on the shared codebook.

\begin{figure}
\begin{center}
\includegraphics[width=\textwidth/2,height=\textheight,keepaspectratio]{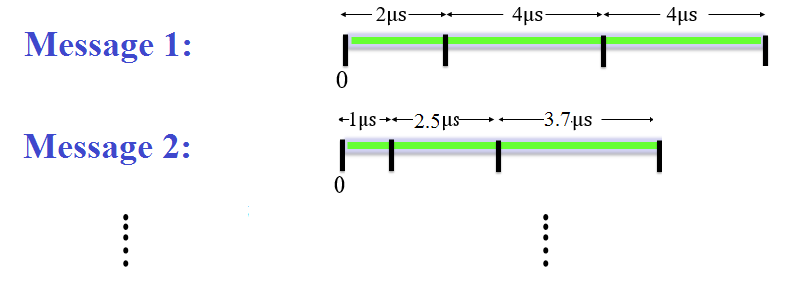}
\end{center}
 \caption{Codebook generation: Alice and Bob share a codebook (secret), which specifies the sequence of inter-packet delays corresponding to each message. Each letter of the codebook is obtained by generating a random variable according to $p_0(x)$.}
\label{fig:codebook}
 \end{figure}

Per above, Alice's communication includes two phases: a buffering phase and a transmission phase. During the buffering phase $[0,\tau_{N \psi}]$, where $0< \psi<1$ is a parameter to be defined later, Alice slows down the packet transmission in order to build up packets in her buffer. In particular, Alice's purpose in the first phase is to buffer enough packets to ensure, with high probability, she will not run out of packets during the transmission phase $(\tau_{N \psi},\tau_{N}]$ (see Fig.~\ref{fig:Twophased2}).

\begin{figure}
\begin{center}
\includegraphics[width=\textwidth/2,height=\textheight,keepaspectratio]{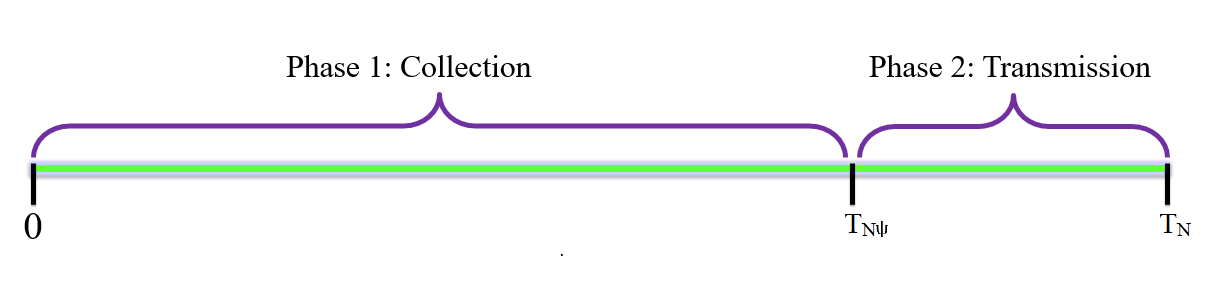}
\end{center}
 \caption{Two-phase construction: Alice's communication includes two phases. In the first phase, Alice slows down the transmission and buffers the excess packets. In the next phase, she transmits packets to Bob according to the inter-packet delays in the codeword corresponding to the message to be transmitted.}
 \label{fig:Twophased2}
 \end{figure}
 
\textit{Analysis}: Since the queue is initially in equilibrium and in both of the phases, Alice's packet transmission rate remains less than the service rate of the queue, the queue stays in equilibrium during the scenario. Thus, we can use the results of Lemma~\ref{lem:3}. 

({\em Covertness}) Suppose that Willie knows when each of the two phases will start and end if Alice chooses to transmit to Bob. Next, we show that during the first phase, Alice's buffering is covert. By Lemma~\ref{lem:2}, Alice can buffer $m=\mathcal{O}\left(\sqrt{N}\right)$ in the first phase where 
\begin{align}
\label{th4m} \lim\limits_{N \to \infty}\mathbb{P}\left(m \geq \epsilon \sqrt{\frac{ N \psi }{4c}}\right)=1,
\end{align}
\noindent while lower bounding the sum of Willie's error probability by $1-\epsilon$ where $0<\epsilon<1$. Thus, Alice's buffering is covert in this phase. 

During the second phase, the packet timings corresponding to the selected codeword are an instantiation of a renewal point process with inter-arrival pdf $p_0(x)$ and hence the traffic pattern is indistinguishable from the pattern that Willie expects on the link from Jack to Steve. Hence, the scheme is covert.

({\em Reliability}) Next, we show that Alice will have a reliable communication to Bob. The notion of reliability is tied to two events. First, Bob should be able to decode the message with arbitrarily low probability of error. This follows by adopting the proposed coding scheme as well as condition~\eqref{q:01} (see Lemma~\ref{lem:3}). Second, Alice needs to avoid a ``failure'' event, in which Alice is unable to create the packet timings for the selected codeword because she has run out of packets in her buffer at some point in the codeword transmission process. 

In the first phase of Scenario 2, Alice uses the same buffering technique on the same number of packets,$N \psi$ , as in that of Scenario 1. Therefore, in both of the scenarios, she can collect $m=\mathcal{O}\left(\sqrt{N}\right)$ packets in the first phase (see~\eqref{th3m} and~\eqref{th4m}). Also, in the second phase of Scenario 4, the rate at which she receives and transmits overt packets is $\lambda$, which is the same as in Scenario 3. Combined with the fact that the second phase in both of the scenarios starts when Alice receives the $(\psi N+1)^{th}$ packet and ends when Alice receives the $N^{th}$ packet, the failure analysis of Scenario 4 follows from the one in Scenario 3 (Theorem 3) and we can show that Alice can achieve $\mathbb{P}_f<\zeta$ for any $\zeta>0$, as long as $\frac{ \psi}{1- \psi} =\left( \frac{2 }{\epsilon} \operatorname{erf}^{-1} \left(1-\zeta\right)\right)^2$, where $\mathbb{P}_f$ is probability of the event ``failure''.

({\em Number of Covert Bits}) By Lemma~\ref{lem:3}, the capacity of the $G/M/1$ queue for conveying information through inter-packet delays is $C(\lambda)>0$ (nats/packet), where $C(\lambda)$ is defined in~\eqref{q:00}. Therefore, she can transmit covertly and reliably $N_b=C(\lambda) \left(\tau_N-\tau_{\psi N}\right) = C(\lambda) \tau_{N(1-\psi)}$ nats to Bob. Since $\tau_N$ is sum of $N$ i.i.d inter-arrival times, the WLLN yeilds $\frac{\tau_{N }}{N} \xrightarrow{P} \lambda^{-1}$. Therefore
\begin{align}
C(\lambda)\frac{\tau_{N (1-\psi)}}{N(1-\psi)}= \frac{N_b}{N(1-\psi)} \xrightarrow{P} C(\lambda) \lambda^{-1}
\end{align}
\noindent Thus, we can easily show that 
\begin{align}
\lim\limits_{N \to \infty}\mathbb{P}\left(N_b \geq \frac{C(\lambda) N \left(1-\psi\right)}{\lambda}\right)=1
\end{align}
\noindent Hence, Alice can send $N_b=\mathcal{O}\left(N\right)$ bits to Bob covertly and reliably. 

({\em Size of the Codebook}) According to Definition~\ref{def:v1}, the rate of the codebook is $\frac{\lambda \log M }{N (1-\psi)}$ where $M$ is the size of the codebook. Since the capacity of the queue $C(\lambda)$ is the maximum achievable rate at output rate $\lambda$ (see Definition~\ref{def:rate}), the size of the codebook is
\begin{align}
M = e^{{\left(1- \psi\right) N C(\lambda) }}
\end{align}
\noindent where $C(\lambda)$ is defined in~\eqref{q:00} and $1-\psi = \left(\left( \frac{2 }{\epsilon} \operatorname{erf}^{-1} \left(1-{\zeta}\right)\right)^2 +1\right)^{-1}$.

Here, in the covertness analysis, calculation of the number of covert bits, and the size of the codebook, we have proved that the transmission is covert given ``failure'' does not occur. Since $\mathbb{P}_f<\zeta$ for any $\zeta>0$, the results hold for the general case too.
\end{proof}

\section{Discussion}
Although the regulatory conditions \eqref{c1}-\eqref{c3} required for Lemma~\ref{lem:2}, Theorems 3 and 5 seem restrictive, many probability distributions satisfy these conditions. For example, the generalized gamma distribution and its special cases, exponential distribution, Chi-squared distribution, Rayleigh distribution, Weibull distribution, Gamma distribution, and Erlang distribution, satisfy \eqref{c1}-\eqref{c3}. 
Among the distributions that do not satisfy conditions \eqref{c1}-\eqref{c3}, are included any distributions whose support is not $[0,\infty)$, such as the Uniform distribution on $[a,b]$. The intuition is that if Alice slows down the packet stream (which results in scaling up the distribution $p_0(x)$) to buffer packets, for a large number of packets, she produces with high probability an inter-packet delay that for certain could not have been generated by $p_0(x)$. Thus, Willie will detect Alice's buffering with high probability.

\section{Conclusion}
We present two scenarios for covert communication on a general (i.e. not necessarily Poisson) renewal channel, hence significantly extending our previous work. In the first scenario, since the packets are not authenticated by adversary Willie, Alice communicates with Bob by insertion of the packets into the channel. We propose a two-phase scheme for Alice. In the first phase, she slows down the packet stream to buffer some packets. In the second phase, she inserts her own packets along with Jack's transmitted packets in a slightly higher rate renewal process. If the total number of transmitted packets from Jack to Steve is $N$, Alice can covertly insert $\mathcal{O}\left(\sqrt{N}\right)$ packets. Next, we analyze the scenario where Willie authenticates the packets; therefore, Alice cannot insert packets. However, we assume that Alice and Bob share a secret key, allowing them to share a secret codebook, and that the only distortion between Alice and Bob is a stable queue. We showed that if Alice buffers some packets first, she can reliably and covertly send $\mathcal{O}\left(N\right)$ bits to Bob if the total number of packets transmitted by Jack is $N$.

\bibliographystyle{ieeetr}

\begin{thebibliography}{10}

\bibitem{snowden}
``Edward {Snowden}: Leaks that exposed {US} spy programme.''
  \url{http://www.bbc.com/news/world-us-canada-23123964}, Jan 2014.

\bibitem{ker07pool}
A.~D. Ker, ``Batch steganography and pooled steganalysis,'' vol.~4437 of {\em
  Lecture Notes in Computer Science}, pp.~265--281, Springer Berlin Heidelberg,
  2007.

\bibitem{bash12sqrtlawisit}
B.~A. Bash, D.~Goeckel, and D.~Towsley, ``Square root law for communication
  with low probability of detection on {AWGN} channels,'' in {\em Proc. {IEEE}
  Int. Symp. Inform. Theory (ISIT)}, (Cambridge, MA, USA), pp.~448--452, July
  2012.

\bibitem{bash_jsac2013}
B.~Bash, D.~Goeckel, and D.~Towsley, ``Limits of reliable communication with
  low probability of detection on {AWGN} channels,'' {\em Selected Areas in
  Communications, IEEE Journal on}, vol.~31, pp.~1921--1930, September 2013.

\bibitem{korzhik2005existence}
V.~Korzhik, G.~Morales-Luna, and M.~H. Lee, ``On the existence of perfect
  stegosystems,'' in {\em International Workshop on Digital Watermarking},
  pp.~30--38, Springer, 2005.


\bibitem{boulat_commmagg}
B.~A. Bash, D.~Goeckel, D.~Towsley, and S.~Guha, ``Hiding information in noise:
  Fundamental limits of covert wireless communication,'' {\em IEEE
  Communications Magazine: Special Issue on Wireless Physical Layer Security},
  Dec. 2015.
\newblock to appear.

\bibitem{che13sqrtlawbscisit}
P.~H. Che, M.~Bakshi, and S.~Jaggi, ``Reliable deniable communication: Hiding
  messages in noise,'' in {\em Proc. {IEEE} Int. Symp. Inform. Theory (ISIT)},
  (Istanbul, Turkey), pp.~2945--2949, July 2013.

\bibitem{hou14isit}
J.~Hou and G.~Kramer, ``Effective secrecy: Reliability, confusion and
  stealth,'' in {\em Proc. {IEEE} Int. Symp. Inform. Theory (ISIT)}, (Honolulu,
  HI, USA), pp.~601--605, July 2014.

\bibitem{bloch15covert-isit}
M.~Bloch, ``Covert communication over noisy memoryless channels: A
  resolvability perspective,'' in {\em Proc. {IEEE} Int. Symp. Inform. Theory
  (ISIT)}, (Hong Kong, China), 2015.

\bibitem{wang15covert-isit}
L.~Wang, G.~W. Wornell, and L.~Zhang, ``Limits of low-probability-of-detection
  communication over a discrete memoryless channel,'' in {\em Proc. {IEEE} Int.
  Symp. Inform. Theory (ISIT)}, (Hong Kong, China), 2015.

\bibitem{soltani2015covert}
R.~Soltani, D.~Goeckel, D.~Towsley, and A.~Houmansadr, ``Covert communications
  on poisson packet channels,'' in {\em 2015 53rd Annual Allerton Conference on
  Communication, Control, and Computing (Allerton)}, pp.~1046--1052, IEEE,
  2015.

\bibitem{verdubitsq}
V.~Anantharam and S.~Verdu, ``Bits through queues,'' {\em Information Theory,
  IEEE Transactions on}, vol.~42, no.~1, pp.~4--18, 1996.

\bibitem{kullback1968information}
S.~Kullback, ``Information theory and statistics,'' 1968.

\bibitem{bash13squarerootjsac}
B.~A. Bash, D.~Goeckel, and D.~Towsley, ``Limits of reliable communication with
  low probability of detection on {AWGN} channels,'' {\em IEEE Journal on
  Selected Areas in Communications}, vol.~31, no.~9, pp.~1921--1930, 2013.

\bibitem{grandell1991aspects}
J.~Grandell, {\em Aspects of risk theory}.
\newblock Springer Science \& Business Media, 1991.

\bibitem{sandhya1991geometric}
E.~Sandhya, ``On geometric infinite divisibility, p-thinning and cox
  processes,'' {\em J. Kerala Statist. Assoc}, vol.~7, pp.~1--10, 1991.

\bibitem{verdu1994general}
S.~Verdu and T.~S. Han, ``A general formula for channel capacity,'' {\em
  Information Theory, IEEE Transactions on}, vol.~40, no.~4, pp.~1147--1157,
  1994.

\bibitem{han1993approximation}
T.~S. Han and S.~Verdu, ``Approximation theory of output statistics,'' {\em
  Information Theory, IEEE Transactions on}, vol.~39, no.~3, pp.~752--772,
  1993.

\bibitem{shiryaev1996probability}
A.~N. Shiryaev, ``Probability, volume 95 of graduate texts in mathematics,''
  1996.

\end{thebibliography}

\appendix
\paragraph{Proof of~\eqref{eq:5}}
Observe that
\begin{align}
  \mathbb{P}\left(N_c \geq \epsilon \sqrt{\frac{N}{4 c}}\right) =  \mathbb{P}\left(N_c \geq \frac{\rho N}{{2}}\right)
 &= \mathbb{P}\left(N-X\left(\tau_N\left(1-\rho\right)\right) \geq  \frac{\rho N}{{2}}\right)\\
  \nonumber &=  \mathbb{P}\left(X\left(\tau_N\left(1-\rho\right)\right) \leq  N\left(1- \frac{\rho}{{2}}\right) \right)\\
    \nonumber &=  \mathbb{P}\left(\tau_N\left(1-\rho\right)\leq \tau_{N\left(1- \frac{\rho }{{2}}\right)}\right)
\end{align}
\noindent where the last step is true since $\mathbb{P}\left( \tau_{i} \leq T\right)= \mathbb{P}\left(N_c \geq i\right)$. Let $A_1, A_2, \cdots$ be the inter-arrivals of the packets transmitted by Jack. Therefore,
\begin{align}
\mathbb{P}\left(N_c \geq \epsilon \sqrt{\frac{N}{ 4 c}}\right) &=  \mathbb{P}\left(\left(1-\rho\right) \sum_{i=1}^{N} A_i \leq \sum_{i=1}^{N (1-\rho/{2})} A_i  \right)\\
 &=  \mathbb{P}\left(\left(1-\rho\right) \sum_{i=N (1-\rho/{2})+1}^{N} A_i \leq \rho \sum_{i=1}^{N (1-\rho/{2})} A_i  \right)\\
 &=  \mathbb{P}\left( \sum_{i=N (1-\rho/{2})+1}^{N} \frac{A_i}{N \rho/2} \leq 2 \frac{N(1-\rho/2)}{N(1-\rho)}  \sum_{i=1}^{N (1-\rho/{2})} \frac{A_i}{N (1-\rho/2)}  \right)
\end{align}
\noindent Let 
\begin{align}
\nonumber A_N^{*} &= \sum_{i=N (1-\rho/{2})+1}^{N} \frac{A_i}{N \rho/2}\\
\nonumber A_N^{**} &=\sum_{i=1}^{N (1-\rho/{2})} \frac{A_i}{N (1-\rho/2)}
\end{align}
\noindent Therefore, 
\begin{align}
\label{apeq1} \mathbb{P}\left(N_c \geq \epsilon \sqrt{\frac{N}{ 4 c}}\right)
 &=  \mathbb{P}\left( A_N^{*}\leq 2 \frac{N(1-\rho/2)}{N(1-\rho)} A_N^{**} \right)\\
 \label{apeq2} &\geq \mathbb{P}\left( A_N^{*}\leq 2 \frac{N(1-\rho)}{N(1-\rho)} A_N^{**} \right)\\
  \label{apeq3}&=\mathbb{P}\left( A_N^{*}\leq 2  A_N^{**} \right)
\end{align}
\noindent where \eqref{apeq2} follows from \eqref{apeq1} since $\{A_N^{*}\leq 2 \frac{N(1-\rho)}{N(1-\rho)} A_N^{**} \}\subset \{A_N^{*}\leq 2 \frac{N(1-\rho/2)}{N(1-\rho)} A_N^{**} \}$. Now, 
by the WLLN, 
\begin{align}
\nonumber A_N^{*} &\xrightarrow{P} \lambda^{-1}\\
\nonumber A_N^{**} &\xrightarrow{P} \lambda^{-1}
\end{align} 
\noindent Since $c_1 A_N^{*} + c_2 A_N^{**} \xrightarrow{P} c_1 \lambda^{-1} + c_2\lambda^{-1}$ for any two real numbers $c_1$ and $c_2$, (see \cite[problem 5 page 262]{shiryaev1996probability}), $2 A_N^{**} - A_N^{*} \xrightarrow{P} \lambda^{-1}$. Therefore, for any $\gamma>0$, $\lim_{N \to \infty }\mathbb{P}\left(|2 A_N^{**} - A_N^{*} - \lambda^{-1}|\leq\gamma\right)=1$. Consequently, 
\begin{align}
\nonumber \lim_{N \to \infty }\mathbb{P}\left(2 A_N^{**} - A_N^{*} - \lambda^{-1}\geq-\gamma\right)=1
\end{align} 
\noindent Let, $\gamma= \lambda^{-1}$. Thus,
\begin{align}
 \label{apeq4} \lim_{N \to \infty }\mathbb{P}\left(2 A_N^{**} - A_N^{*}  \geq 0\right)=1
\end{align} 
\noindent Therefore, by~\eqref{apeq3} and~\eqref{apeq4} 
\begin{align}
 \nonumber &\lim_{N \to \infty} \mathbb{P}\left(N_c \geq \epsilon \sqrt{\frac{N}{ 4 c}}\right)=1
\end{align}
\paragraph{Proof of~\eqref{c4}} This is true according to [Ch. 2.6]\cite{kullback1968information}, and $c$ is the Fisher information which is given by
\begin{align}
\label{eq:ap4}c=\int_{x=0}^{\infty} p_0(x) \frac{1}{p_0(x)^2} \left(\frac{\partial p_1(x,\rho)}{\partial \rho}\bigg|_{\rho=0}\right)^2 dx
\end{align}
\noindent Since $p_1(x,\rho)=(1-\rho)p_0(x(1-\rho))$,
\begin{align}
 \frac{\partial p_1(x,\rho)}{\partial \rho} \bigg|_{\rho=0} &= \frac{\partial \left(\left(1-\rho\right) p_0\left(\left(1-\rho\right)x\right)\right)}{\partial \rho} \bigg|_{\rho=0}
\label{eq:ap6} =-p_0(x) - x \frac{d p_0(x)}{dx}
\end{align}
\noindent Therefore,~\eqref{eq:ap4} becomes
\begin{align}
 c&=\int_{x=0}^{\infty} p_0(x) + 2x \frac{d p_0(x)}{dx} +  \frac{x^2}{p_0(x)} \left(\frac{d p_0(x)}{dx}\right)^2 dx
\label{eq:ap5} =1+ \int_{x=0}^{\infty} 2x \frac{d p_0(x)}{dx} +  \frac{x^2}{p_0(x)} \left(\frac{d p_0(x)}{dx}\right)^2 dx
\end{align}
\noindent Consider $ 2x \frac{d p_0(x)}{dx}$ in the above equation.  By~\eqref{c3}, $\int_{x=0}^{\infty}\frac{\partial p_1(x,\rho)}{\partial \rho}\big|_{\rho=0}dx=0$. Therefore, by~\eqref{eq:ap6}
\begin{align}
\nonumber \int_{x=0}^{\infty}p_0(x) + x \frac{d p_0(x)}{dx} dx  = 0
\end{align}
\noindent Consequently
\begin{align}
\nonumber \int_{x=0}^{\infty}  x \frac{d p_0(x)}{dx} dx  = - \int_{x=0}^{\infty} p_0(x)  dx  =-1
\end{align}
\noindent Thus,~\eqref{eq:ap5} becomes
\begin{align}
\nonumber c&=-1+ \int_{x=0}^{\infty}  \frac{x^2}{p_0(x)} \left(\frac{d p_0(x)}{dx}\right)^2 dx =-1+ \int_{x=0}^{\infty}  p_0(x) {x^2} \left(\frac{d \log p_0(x)}{dx}\right)^2 dx
\end{align}

\paragraph{Proof of~\eqref{lem2pmd}} 
\begin{align}
\nonumber \mathbb{P}_{MD}=\mathbb{P}\left(S\leq N' /\lambda + U \bigg| H_1\right) &=\mathbb{P}\left(S -  \frac{N'}{\lambda(1-\rho)} \leq \frac{N'}{\lambda} + U - \frac{N'}{\lambda(1-\rho)}  \bigg| H_1\right)\\
\nonumber &=\mathbb{P}\left(S -  \frac{N'}{\lambda(1-\rho)} \leq \frac{N'}{\lambda}  \frac{\rho}{\rho-1}+ U   \bigg| H_1\right)\\
\nonumber &=\mathbb{P}\left(-\left(S -  \frac{ N'}{\lambda(1-\rho)} \right) \geq -\left( \frac{N'}{\lambda} \frac{\rho}{\rho-1}+ U  \right) \bigg| H_1\right)\\
\label{eq:PMDupperbound12} &\leq\mathbb{P}\left(\Big|S -  \frac{N' }{\lambda(1-\rho)} \Big| \geq \Big| \frac{N'}{\lambda} \frac{\rho}{\rho-1}+ U  \Big| \bigg| H_1\right)
\end{align}
\noindent Therefore, applying Chebyshev's inequality on~\eqref{eq:PMDupperbound12} yields 
\begin{align} 
 \mathbb{P}_{MD}\leq \frac{N' \frac{ \sigma^2}{\left(1-\rho\right)^2}}{\left( \frac{N'}{\lambda} \frac{\rho}{\rho-1}+ U \right)^2}
= \frac{N'  \sigma^2}{\left(1-\rho\right)^2 \left( \frac{N'}{\lambda}  \frac{\rho}{\rho-1}+ U \right)^2}
&= \frac{N'  \sigma^2}{ \left( \frac{N'}{\lambda} {\rho}+ U \left(1-\rho\right)\right)^2}
\end{align}
\noindent By \eqref{con1u}
\begin{align} 
 \mathbb{P}_{MD}\leq  \frac{N'  \sigma^2}{ \left( N' {\rho}/\lambda + \sqrt{\frac{N' }{\lambda \alpha}} \left(1-\rho\right)\right)^2}
\label {conv2} &= \frac{ \sigma^2}{ \left( \sqrt{N'} {\rho}/\lambda+ {\frac{ 1}{\sqrt{\lambda \alpha}}} \left(1-\rho\right)\right)^2}
\end{align}
\noindent Consider the denominator of~\eqref{conv2}. Since $\sqrt{N'} {\rho}/\lambda>0$ and $\frac{ 1}{\sqrt{\lambda \alpha}}>0$,
\begin{align} 
\nonumber \mathbb{P}_{MD}&\leq  \frac{ \sigma^2}{ \left( \sqrt{N'} {\rho}/\lambda\right)^2}
\end{align}

\paragraph{Proof of~\eqref{th3e6}} 
\noindent By~\eqref{th3:1}, 
\begin{align}
\label{th3e2}\mathbb{E} [\sqrt{\mathcal{D}(\mathbb{P}_0 || \mathbb{P}_1)}] \leq \mathbb{E}[\sqrt{N_{oc}}] \sqrt{\mathcal{D}\left({p}_0(x) || {p}_1(x)\right)}
\end{align}
\noindent where $\mathbb{E}[\cdot]$ denotes expectation over all possible values of the random variable $N_{oc}$. By the Law of Total Expectation
\begin{align}
\mathbb{E}[\sqrt{N_{oc}}] =  
 \nonumber &\mathbb{E}[\sqrt{N_{oc}}|N_{oc} \leq 2N\left(1-\psi\right)] \mathbb{P} \left( N_{oc} \leq 2N\left(1-\psi\right)\right)\\
\nonumber  &\phantom{=}+   \mathbb{E}[\sqrt{N_{oc}}|N_{oc} > 2N\left(1-\psi\right)] \mathbb{P} \left( N_{oc} > 2N\left(1-\psi\right)\right)\\
\label{th3e1} &\leq \mathbb{E}[\sqrt{N_{oc}}|N_{oc} \leq 2N\left(1-\psi\right)] + \mathbb{P} \left( N_{oc} > 2N\left(1-\psi\right)\right)
\end{align}
 \noindent Consider $\mathbb{E}[\sqrt{N_{oc}}|N_{oc} \leq 2N\left(1-\psi\right)]$ in~\eqref{th3e1}. 
 \begin{align}
\label{th3e3} \mathbb{E}[\sqrt{N_{oc}}|N_{oc} \leq 2N\left(1-\psi\right)] \leq \sqrt{2N\left(1-\psi\right)}
 \end{align}
\noindent Now, consider $\mathbb{P} \left( N_{oc} > 2N\left(1-\psi\right)\right)$ in~\eqref{th3e1}. Similar to the arguments that leads to~\eqref{eq:5} and~\eqref{NocN1} we can show that
\begin{align}
\label{NocN2} \lim_{N \to \infty}\mathbb{P}\left(N_{oc} > {2 N \left(1-\psi\right)}\right) &=0
\end{align}
\noindent Hence, by~\eqref{th3e1},~\eqref{th3e3},~\eqref{NocN2}
\begin{align}
\label{th3:4}\lim\limits_{N \to \infty} \mathbb{E}[\sqrt{N_{oc}}] \leq \sqrt{2 N \left(1-\psi\right)}
\end{align}
\noindent Therefore, by~\eqref{th3e2} and~\eqref{th3:4}
\begin{align}
\label{th3e4}\lim\limits_{N \to \infty} \mathbb{E} [\sqrt{\mathcal{D}(\mathbb{P}_0 || \mathbb{P}_1)}] \leq  \sqrt{2 N (1-\psi)} \sqrt{\mathcal{D}\left({p}_0(x) || {p}_1(x)\right)}
\end{align}
\noindent Recall that according to~\eqref{th3rhoVal},  $\rho=  \frac{\epsilon}{\sqrt{{2 c N (1-\psi)}}}$ and therefore, $\rho \to 0$ as $N \to \infty$. Hence, by~\eqref{dklkhaf2} and~\eqref{th3e4}
\begin{align}
\label{th3e5}\lim\limits_{N \to \infty} \mathbb{E} [\sqrt{\mathcal{D}(\mathbb{P}_0 || \mathbb{P}_1)}] \leq  \lim\limits_{N \to \infty} \sqrt{2 N (1-\psi) c \rho^2}=\sqrt{2 N (1-\psi) c \frac{\epsilon^2}{2 c N (1-\psi)}}= \epsilon
\end{align}
\noindent Thus, 
\begin{align}
\label{th3easdf6}\lim\limits_{N \to \infty} \mathbb{E} [\sqrt{\frac{\mathcal{D}(\mathbb{P}_0 || \mathbb{P}_1)}{2}}] &\leq \frac{\epsilon}{\sqrt{2}} < \epsilon
\end{align}

\paragraph{Proof of~\eqref{randw1}} According to~\eqref{th3K}, $K=N_o+N(1-\psi)$. Therefore,
\begin{align}
 \lim\limits_{N \to \infty}  \mathbb{P}\left(K\geq 4 N\left(1-\psi\right)\right) = \lim\limits_{N \to \infty}  \mathbb{P}\left(N_o\geq 3 N\left(1-\psi\right)\right)
\label{firsteq}&= \lim\limits_{N \to \infty}  \mathbb{P}\left(\frac{N_o}{N\left(1-\psi\right)}\geq 3\right)
\end{align}
\noindent Note that there is a symmetry between the total number of overt packets $N_o$ and the total number of covert packets $N_c$ that Alice transmits in the second phase. Observe 
\begin{align}
N_o = \sum\limits_{i=1}^{N_{oc}} (1-b_i)
\end{align}
\noindent Therefore, 
\begin{align}
\label{step1}  \mathbb{P}\left(K\geq 4 N\left(1-\psi\right)\right) &=  \mathbb{P}\left(\sum\limits_{i=1}^{N_{oc}} \frac{1-b_i}{N\left(1-\psi\right)} \geq 3\right)\\
\label{step2} &\leq  \mathbb{P}\left(\sum\limits_{i=1}^{N_{oc}} \frac{1}{N\left(1-\psi\right)} \geq 3\right)\\
\label{step3} &=  \mathbb{P}\left( N_{oc} \geq 3 {N\left(1-\psi\right)}\right)
\end{align}
\noindent where \eqref{step2} follows from \eqref{step1} since each of the $b_i$s corresponds to an outcome of a Bernoulli process therefore $b_i \leq 1$ and consequently $\Big\{\sum\limits_{i=1}^{N_{oc}} \frac{1-b_i}{N\left(1-\psi\right)} \geq 3\Big\} \subset \Big\{\sum\limits_{i=1}^{N_{oc}} \frac{1}{N\left(1-\psi\right)} \geq 3\Big\}$. Now, by~\eqref{NocN2}, $\lim_{N \to \infty}\mathbb{P}\left(N_{oc} > {2 N \left(1-\psi\right)}\right) =0$. Therefore, 

\begin{align}
\label{step4} \lim\limits_{N \to \infty}\mathbb{P}\left( N_{oc} \geq 3 {N\left(1-\psi\right)}\right)=0
\end{align}
\noindent Thus, by \eqref{step3} and \eqref{step4} 
\begin{align}
\nonumber  \lim\limits_{N \to \infty} \mathbb{P}\left(K\geq 4 N\left(1-\psi\right)\right) &=0
\end{align}

\paragraph{Proof of~\eqref{q:10}} 
\begin{align}
 \frac{\mathbb{Q}_{D^n|A^n}\left(d^n|a^n\right)}{\mathbb{Q}_{D^n}\left(d^n\right)}  
\label{q:1} =& \frac{\mathbb{P}_{D^n|A^n}\left(d^n|a^n\right)}{\mathbb{P}_{D^n}\left(d^n\right)} \frac{\mathbb{P}_{D^n}\left(d^n\right)}{\mathbb{Q}_{D^n}\left(d^n\right)} \frac{\mathbb{Q}_{D^n|A^n}\left(d^n|a^n\right)}{\mathbb{P}_{D^n|A^n}\left(d^n|a^n\right)}\\
\label{q:2} =& \frac{\mathbb{P}_{D^n|A^n}\left(d^n|a^n\right)}{\mathbb{P}_{D^n}\left(d^n\right)} \frac{\mathbb{P}_{D^n}\left(d^n\right)}{\mathbb{Q}_{D^n}\left(d^n\right)} \frac{\mathbb{Q}_{D^n,A^n}\left(d^n,a^n\right)}{\mathbb{P}_{D^n,A^n}\left(d^n,a^n\right)}
\frac{\mathbb{P}_{A^n}\left(a^n\right)}{\mathbb{Q}_{A^n}\left(a^n\right)}=\\
\label{q:3} &  \frac{\mathbb{P}_{D^n|A^n}\left(d^n|a^n\right)}{\mathbb{P}_{D^n}\left(d^n\right)}  \frac{\mathbb{Q}_{A^n|D^n}\left(a^n|d^n\right)}{\mathbb{P}_{A^n|D^n}\left(a^n|d^n\right)}
\frac{\mathbb{P}_{A^n}\left(a^n\right)}{\mathbb{Q}_{A^n}\left(a^n\right)}
\end{align}

\noindent where \eqref{q:2} follows from \eqref{q:1} since $\frac{\mathbb{Q}_{D^n|A^n}\left(d^n|a^n\right)}{\mathbb{P}_{D^n|A^n}\left(d^n|a^n\right)} =  \frac{\mathbb{Q}_{D^n,A^n}\left(d^n,a^n\right)}{\mathbb{P}_{D^n,A^n}\left(d^n,a^n\right)}
\frac{\mathbb{P}_{A^n}\left(a^n\right)}{\mathbb{Q}_{A^n}\left(a^n\right)}$ and \eqref{q:3} follows from \eqref{q:2} since $\frac{\mathbb{P}_{D^n}\left(d^n\right)}{\mathbb{Q}_{D^n}\left(d^n\right)} \frac{\mathbb{Q}_{D^n,A^n}\left(d^n,a^n\right)}{\mathbb{P}_{D^n,A^n}\left(d^n,a^n\right)}
 =\frac{\mathbb{Q}_{A^n|D^n}\left(a^n|d^n\right)}{\mathbb{P}_{A^n|D^n}\left(a^n|d^n\right)} $. Therefore
 
 \begin{align}
 \frac{1}{n}\log \frac{\mathbb{Q}_{D^n|A^n}\left(d^n|a^n\right)}{\mathbb{Q}_{D^n}\left(d^n\right)}= \frac{1}{n} \log{\frac{\mathbb{P}_{D^n|A^n}\left(d^n|a^n\right)}{\mathbb{P}_{D^n}\left(d^n\right)}}+\frac{1}{n} \log{\frac{\mathbb{Q}_{A^n|D^n}\left(a^n|d^n\right)}{\mathbb{P}_{A^n|D^n}\left(a^n|d^n\right)}
 }+\frac{1}{n}\log{\frac{\mathbb{P}_{A^n}\left(a^n\right)}{\mathbb{Q}_{A^n}\left(a^n\right)}}
 \end{align}

\paragraph{Proof of~\eqref{q:6}} 
Similar to the arguments that yield \cite[Eq 2.25]{verdubitsq}, we can show that 

\begin{align}
\label{q:16}\frac{1}{n} \log{\frac{\mathbb{P}_{D^n|A^n}\left(d^n|a^n\right)}{\mathbb{P}_{D^n}\left(d^n\right)}}\nonumber &= \log{\frac{\mu}{\lambda}}+\frac{\lambda-\mu}{n}\sum\limits_{i=1}^{n} d_i + \frac{\mu}{n}\sum\limits_{i=1}^{n} w_i-\frac{1}{n} i_{D_0;D_1,\cdots,D_n}\left(d_0;d_1,\cdots,d_n\right)
\end{align}

\noindent where 
\begin{align}
 & i_{D_0;D_1,\cdots,D_n}\left(d_0;d_1,\cdots,d_n\right) = \log{\frac{\mathbb{P}_{D_0|D_1,\cdots,D_n}\left(d_0|d_1,\cdots,d_n\right)}{\mathbb{P}_{D_0}\left(d_0\right)}}
\end{align}

\noindent Consider $\frac{\lambda-\mu}{n}\sum\limits_{i=1}^{n} d_i +\frac{\mu}{n}\sum\limits_{i=1}^{n} w_i$

\begin{align}
\label{q:12}&\frac{\lambda-\mu}{n}\sum\limits_{i=1}^{n} d_i +\frac{\mu}{n}\sum\limits_{i=1}^{n} w_i =\\ \label{q:13}&\frac{\lambda}{n}\sum\limits_{i=1}^{n} d_i -\frac{\mu}{n}\sum\limits_{i=1}^{n} s_i\\
\label{q:14}&\xrightarrow{P}\lambda \frac{1}{\lambda} - \mu \frac{1}{\mu}=0
\end{align}

\noindent where \eqref{q:13} follows from \eqref{q:12} since $w_i=d_i-s_i$ and \eqref{q:14} follows from \eqref{q:13} because of the WLLN and the fact that the output rate and the service rate of the $G/M/1$ queue is $\lambda$ and $\mu$ respectively. Thus, 

\begin{align}
\label{q:15}\frac{\lambda-\mu}{n}\sum\limits_{i=1}^{n} d_i +\frac{\mu}{n}\sum\limits_{i=1}^{n} w_i \xrightarrow{P} 0
\end{align}

Now, consider $i_{D_0;D_1,\cdots,D_n}\left(d_0;d_1,\cdots,d_n\right)$. Similar to the arguments in \cite[Lemma 1]{verdubitsq}, we can show that 
\begin{align}
\label{q:20} i_{D_0;D_1,\cdots,D_n}\left(d_0;d_1,\cdots,d_n\right) \xrightarrow{P} 0
\end{align}

Note that if $Z_n^1\xrightarrow{P} Z^1$, $Z_n^2\xrightarrow{P} Z^2$, where $Z_n^1,Z_n^2$ are sequences of random variables and $Z^1,Z^2$ are random variables, then $Z_n^1 + Z_n^2\xrightarrow{P} Z^1+Z^2$ \cite[problem 5, p 262]{shiryaev1996probability}. Therefore, \eqref{q:16}, \eqref{q:15}, \eqref{q:20} yield 

\begin{align}
\frac{1}{n} \log{\frac{\mathbb{P}_{D^n|A^n}\left(d^n|a^n\right)}{\mathbb{P}_{D^n}\left(d^n\right)}} \xrightarrow{P} \log{\frac{\mu}{\lambda}}
\end{align}

\noindent Consequently,~\eqref{q:6} holds.

\paragraph{Proof of~\eqref{q:7}} 
Similar to the arguments in \cite[p. 13]{verdubitsq}, we can show that

\begin{align}
\nonumber  &\mathbb{P}\left(\frac{1}{n} \log{\frac{\mathbb{Q}_{A^n|D^n}\left(a^n|d^n\right)}{\mathbb{P}_{A^n|D^n}\left(a^n|d^n\right)}} < -\gamma/3\right)\leq e^{-\gamma n/3}
\end{align}

\noindent Therefore, 

\begin{align}
\nonumber \lim\limits_{n \to \infty}  &\mathbb{P}\left(\frac{1}{n} \log{\frac{\mathbb{Q}_{A^n|D^n}\left(a^n|d^n\right)}{\mathbb{P}_{A^n|D^n}\left(a^n|d^n\right)}} < -\gamma/3\right)=0
\end{align}

\paragraph{Proof of~\eqref{q:8}} 
Since the inter-arrival times input processes for a $M/M/1$ queue and a $G/M/1$ queue are both independent and identically distributed,

\begin{align}
\nonumber \frac{1}{n}\log{\frac{\mathbb{P}_{A^n}\left(a^n\right)}{\mathbb{Q}_{A^n}\left(a^n\right)}} &= \frac{1}{n}\log{\frac{\prod_{i=1}^{n}  e_\lambda\left(a_i\right)}{\prod_{i=1}^{n}  p_0\left(a_i\right)}} =\frac{1}{n} \sum_{i=1}^{n} \log{\frac{e_{\lambda}\left(a_i\right)}{p_0\left(a_i\right)}}
\end{align}
 
\noindent Therefore, by the SLLN, we can show that 

\begin{align}
\frac{1}{n}\log{\frac{\mathbb{P}_{A^n}\left(a^n\right)}{\mathbb{Q}_{A^n}\left(a^n\right)}} \nonumber \xrightarrow{P} \mathbb{E}_{\mathbb{Q}_A}\left[\log{\frac{e_{\lambda}\left(x\right)}{p_0\left(x\right)}}\right]= \int\limits_{0}^{\infty} p_0(x) \log{\frac{e_{\lambda}\left(x\right)}{p_0\left(x\right)}} dx
  =- \mathcal{D}\left(p_0\left(x\right)||e_{\lambda}\left(x\right)\right)
\end{align}
\noindent Thus,
\begin{align}
 \frac{1}{n}\log{\frac{\mathbb{P}_{A^n}\left(a^n\right)}{\mathbb{Q}_{A^n}\left(a^n\right)}}  &\xrightarrow{P}- \mathcal{D}\left(p_0\left(x\right)||e_{\lambda}\left(x\right)\right)
\end{align}

\noindent Consequently,~\eqref{q:8} holds.
\end{document}